\documentclass[11pt,english,hyperfootnotes=false]{article}
\usepackage[T1]{fontenc}
\usepackage{geometry}
\usepackage{color}
\usepackage{babel}
\usepackage{mathtools}
\usepackage{enumitem}
\usepackage{amsmath}
\usepackage{amsthm}
\usepackage{amssymb}
\usepackage{stackrel}
\usepackage{setspace}
\usepackage{esint}
\usepackage[authoryear]{natbib}
\usepackage[unicode=true,pdfusetitle,
 bookmarks=true,bookmarksnumbered=false,bookmarksopen=false,
 breaklinks=true,pdfborder={0 0 1},backref=false,colorlinks=true]
 {hyperref}
\hypersetup{
 linkcolor=blue, urlcolor=black}

\makeatletter
\theoremstyle{plain}
    \ifx\thechapter\undefined
      \newtheorem{assumption}{\protect\assumptionname}
    \else
      \newtheorem{assumption}{\protect\assumptionname}[chapter]
    \fi
\theoremstyle{remark}
    \ifx\thechapter\undefined
      \newtheorem{rem}{\protect\remarkname}
    \else
      \newtheorem{rem}{\protect\remarkname}[chapter]
    \fi
\theoremstyle{plain}
    \ifx\thechapter\undefined
	    \newtheorem{thm}{\protect\theoremname}
	  \else
      \newtheorem{thm}{\protect\theoremname}[chapter]
    \fi
\theoremstyle{plain}
    \ifx\thechapter\undefined
      \newtheorem{lem}{\protect\lemmaname}
    \else
      \newtheorem{lem}{\protect\lemmaname}[chapter]
    \fi

\usepackage[bottom]{footmisc}
\makeatletter
\makeatletter
\renewcommand{\subparagraph}{%
  \@startsection{subparagraph}{5}%
  {\z@}{3pt}{-1em}%
  {\indent\normalfont\normalsize\bfseries}%
}
\makeatother
\allowdisplaybreaks
\usepackage{etoolbox}
\apptocmd\normalsize{%
\abovedisplayskip=5pt
 \belowdisplayskip=6pt
}{}{}

\newtheorem{manualtheoreminner}{Theorem}
\newenvironment{manualtheorem}[1]{%
  \renewcommand\themanualtheoreminner{#1}%
  \manualtheoreminner
}{\endmanualtheoreminner}

\usepackage[utf8]{inputenc}  
\usepackage{xfrac}  

\newcommand{\indep}{\raisebox{0.05em}{\rotatebox[origin=c]{90}{$\models$}}}

\let\originalleft\left
\let\originalright\right
\renewcommand{\left}{\mathopen{}\mathclose\bgroup\originalleft}
\renewcommand{\right}{\aftergroup\egroup\originalright}

\makeatother

\providecommand{\assumptionname}{Assumption}
\providecommand{\lemmaname}{Lemma}
\providecommand{\remarkname}{Remark}
\providecommand{\theoremname}{Theorem}

\begin{document}
\title{A Theory Guide to Using Control Functions\\
to Instrument Hazard Models\thanks{I thank Christopher Palmer for his thoughtful advice and mentorship,
as well for introducing me to this topic through his job market paper.
At his request, I wrote this theory guide to serve as a complement
to Palmer (2023), which contains an empirical guide as well as a summary
of the theory introduced in this paper. I also thank Whitney Newey
and Jiulei Zhu for insightful comments and feedback. A previous version
of this paper was titled \textquotedblleft A Theory Guide: Instrumenting
a Discrete-Data Proportional Hazards Model with Control Functions.\textquotedblright{}
First version: September 2023.}\bigskip{}
}
\author{William Liu\thanks{Pre-doctoral research assistant at MIT/NBER, MIT Sloan, and Harvard;
\textsf{\protect\href{mailto:liuw@mit.edu}{liuw@mit.edu}}}\bigskip{}
}
\date{November 2023\bigskip{}
}
\maketitle
\begin{abstract}
\begin{singlespace}
{\normalsize{}I develop the theory around using control functions
to instrument hazard models, allowing the inclusion of endogenous
(e.g., mismeasured) regressors. Simple discrete-data hazard models
can be expressed as binary choice panel data models, and the widespread
Prentice and Gloeckler (1978) discrete-data proportional hazards model
can specifically be expressed as a complementary log-log model with
time fixed effects. This allows me to recast it as GMM estimation
and its instrumented version as sequential GMM estimation in a Z-estimation
(non-classical GMM) framework; this framework can then be leveraged
to establish asymptotic properties and sufficient conditions. Whilst
this paper focuses on the Prentice and Gloeckler (1978) model, the
methods and discussion developed here can be applied more generally
to other hazard models and binary choice models. I also introduce
my Stata command for estimating a complementary log-log model instrumented
via control functions (available as }\texttt{\textbf{ivcloglog}}{\normalsize{}
on SSC), which allows practitioners to easily instrument the Prentice
and Gloeckler (1978) model.}\\
\textbf{\normalsize{}}\\
\textbf{\normalsize{}Keywords}{\normalsize{}: Control Function, 2SRI,
Instrumental Variables, Hazard Model, Survival Model, Duration Model,
Proportional Hazards, Complementary Log-Log Model, Econometric Software,
Stata}{\normalsize\par}
\end{singlespace}

\noindent \textbf{\normalsize{}JEL Classification}{\normalsize{}:
C35, C36, C41, C87}{\normalsize\par}
\end{abstract}
\thispagestyle{empty} 

\newpage\setcounter{page}{1} 

\begingroup
\renewcommand*{\arraystretch}{1.5}

\section{Introduction\label{sec:intro}}

Hazard models (also known as survival models or duration models) specify
the rate that events occur conditional on the process history (i.e.,
which events have occurred or not occurred and when). They are most
commonly used in fields such as biostatistics, but have also found
purchase in economics and finance, such as in household finance, labor
economics, and health economics. Hazard models may also prove highly
useful to the rapidly growing new field of genoeconomics -- the use
of genes as instruments (also known as \textit{Mendelian randomization}
in biostatistics) to study economic contexts with instrumental variable
methods. This modeling approach underpins causal inference for phenomena
involving stochastic state transitions. Such phenomena include those
for which the randomness is inherent -- e.g., the duration of unemployment
spells is intrinsically random as a result of search-and-matching
-- and those for which there are unobserved factors modeled as random
shocks -- e.g., loan default can be modeled as a deterministic decision,
but one that is made based on random factors observed by the agent
but not the researcher. (See Kiefer, 1988, for some more examples
of economic applications.)

Like many models however, hazard models can suffer from endogeneity,
which can cause severe inconsistency and bias in parameter estimates.
This is a common problem in practice -- endogeneity can arise due
to measurement error in the regressors,\footnote{Strictly speaking, measurement error in the regressors is a form of
endogeneity. This is because the measurement error, with which the
mismeasured regressor is correlated, ends up in the composite error
term -- introducing measurement error in a regressor therefore causes
the mismeasured regressor to be correlated with the new, composite
error term.} simultaneity, and omitted variables. Unfortunately, instrumental
variable (IV) solutions to this endogeneity problem are non-trivial
when the instrumented model is nonlinear. The naive generalized version
of 2SLS, \textit{2SPS} -- two-stage predictor substitution, which
consists of replacing endogenous regressors with their predicted values,
is not generally valid in a nonlinear model. In this paper, I present
a solution to this problem that uses the control function approach
(also known in statistics as \textit{2SRI} -- two-stage residuals
inclusion), which can be thought of as an alternative generalization
of 2SLS.\footnote{The control function approach in a linear model is actually exactly
equivalent to 2SLS when the same first-stage regressions are used.}

I focus on the use of the control function approach in a popular \emph{discrete-data}
analog of the continuous-time Cox proportional hazards model -- the
Prentice and Gloeckler (1978) model -- though, the methods and discussion
developed here can easily be applied to other hazard models (and binary
choice models).\footnote{By ``discrete data'', I mean both \emph{grouped-time} data -- time-discretized
data, such as monthly panel data -- as well as true \emph{discrete-time}
data -- e.g, data on the cycles of some process.} Proportional hazards models are the most commonly used hazard models
because they can be \emph{distribution-free} (i.e., where the baseline
hazard is not assumed to follow any particular distribution) yet simple
to estimate.\footnote{Other common types of models are additive, accelerated failure time,
and proportional odds hazard models.} On the other hand, discrete-data hazard models are much less common
than continuous-time models, but are important because they appropriately
model the discrete nature of discrete data, which is by far the most
common format of data in economics and many other areas. In contrast,
continuous-time models that adjust for ties are just approximations
to discrete-data hazard models; these inherently introduce misspecification
bias. Obviating this misspecification bias motivates my focus on discrete-data
hazard modeling.

The structure of the paper is as follows. In Section \ref{sec:lit},
I contextualize the paper by explaining its place in and contribution
to both the literature on hazard models with endogenous regressors
as well as the literature on the control function approach. Section
\ref{sec:theory} presents an introduction to the model, a discussion
of the asymptotic theory (including the interpretation and selection
of the control function as well as robustness/asymptotic inefficiency),
and extensions of the core asymptotic theory. In Section \ref{sec:stata},
I introduce and explain my Stata command, \texttt{\textbf{\small{}ivcloglog}},
which allows practitioners to easily instrument the Prentice and Gloeckler
(1978) model via the control function approach. Finally, I conclude
in Section \ref{sec:concl} by summarizing the the main findings of
the paper. Appendix \ref{sec:proof} contains a theoretical proof
of Theorem \ref{thm:CAN}, and Appendix \ref{sec:primer} contains
a primer on grouped-time survival analysis.

\section{Literature Review\label{sec:lit}}

In this section, I first present a survey of the endogenous regressor
hazard model literature and then explain how my paper contributes
to the econometrics and statistics literatures.

\subsubsection*{Hazard Models with Endogenous Regressors}

Endogeneity is a familiar problem for statistical analyses, and the
problems that it can create for hazard models are well-known (Prentice,
1982). Yet, the literature for non-additive hazard modeling often
assumes exogenous regressors,\footnote{Additive hazard models, of course, are much easier to handle because
they can be expressed as linear regression models, and therefore can
be instrumented using 2SLS. For example, see J.\ Li et al.\ (2015),
Martinussen, Vansteelandt, et al.\ (2017), and Tchetgen Tchetgen
et al.\ (2015). (Also see Ying et al., 2019, for a paper that uses
control functions instead.)} and works that do allow for endogenous regressors typically only
allow for mismeasured regressors, e.g., Y.\ Li and Ryan (2004), Huang
and C. Y. Wang (2000, 2006, 2018), and Song and C.\ Y.\ Wang (2014).\footnote{In such contexts, repeated observations of the mismeasured regressors
are valid instruments. This is because, although they are subject
to measurement error, the measurement error terms for each of the
repeated observations are mutually independent (by assumption).} Works that allow for more general endogenous regressors -- ``true''
instrumental variable methods that go beyond just mismeasured regressors
-- are much less common. Some of these are: Bijwaard (2008), Kianian
et al.\ (2021), Loeys and Goetghebeur (2003), Loeys et al.\ (2005),
Martinussen, S{\o}rensen, and Vansteelandt (2019), and L.\ Wang
et al.\ (2022) -- see Choi and O'Malley (2017) for an overview of
a bunch of IV methods. However, none of these provide a solution comparable
in usefulness to that of 2SLS in linear models.\footnote{Another important (but less critical) point is that all of the mentioned
works only consider continuous-time models and are therefore not suitable
for discrete data.}

Many such papers make highly restrictive assumptions, severely reducing
their applicability. This often entails requiring that the instruments
and/or the treatments are binary variables. Indeed, Loeys and Goetghebeur
(2003) and Loeys et al.\ (2005) require both to be binary, and L.\
Wang et al.\ (2022) and Kianian et al.\ (2021) require the instruments
to be binary. The estimators in such works can also have some poor
properties. For example, consider Kianian et al.\ (2021), who apply
\emph{IPW} to survival analysis.\footnote{IPW stands for inverse probability weighting, inverse propensity weighting,
or inverse propensity-score weighting.} Developed by Horvitz and Thompson (1952), IPW is where the estimated
probability of each treatment status (``propensity score'') is used
to construct weights for observations, allowing the estimation of
statistics corresponding to a (pseudo-)population with certain characteristics.
For example, an average treatment effect can be calculated by comparing
the outcomes of treated and untreated populations. The basic (i.e.,
\emph{singly robust}) IPW methods involve using inverted propensity
scores as weights. Trivially, this means that they are very sensitive
to estimated propensity scores being close to or exactly zero --
this is the sample equivalent of the common support assumption being
violated. Propensity scores close to zero result in extremely large
weights, causing instability in the IPW estimates. This problem is
pervasive enough to have motivated the development of \emph{stabilized
IPW}, where the weights are normalized or regularized (Robins, Hern\'an,
Brumback, 2000), and \emph{doubly robust} methods (i.e., AIPW --
\emph{augmented IPW}), where \emph{outcome regression} (i.e., ``regular''
regression methods) is combined with IPW. (See Kang \& Schafer, 2007,
for an overview of the latter.) Unfortunately, Kianian et al.\ (2021)
use the Abadie (2003) IPW-type estimator, which is not stabilized,
hampering their estimator's usefulness.

Much less restrictive proposals are given by Bijwaard (2008) and Martinussen,
S{\o}rensen, and Vansteelandt (2019). Martinussen, S{\o}rensen,
and Vansteelandt (2019) seem to have independently discovered essentially
the same idea as Bijwaard (2008), so I will focus on discussing only
the latter. Bijwaard (2008) is a seminal work that introduces his
\emph{IVLR} (``Instrumental Variable Linear Rank'') estimator; an
excellent summary is Bijwaard (2009).\footnote{Specifically, the IVLR estimator covers the \emph{GAFT} (generalized
accelerated failure time) model, which nests as special cases the
\emph{AFT} (accelerated failure time) model as well as the Cox model
and the \emph{MPH} (mixed proportional hazard) model. (The MPH model
is just a Cox model with ``frailty'' terms -- ``random effects''
in economics terminology.)} ``Linear rank'' refers to a rank test, an equivalent of the score
test for coefficient significance in the Cox model. Ranks pop up here
because they appear in the partial likelihood for a Cox model and
in the likelihoods for other continuous-time models. The idea behind
IVLR is that the instruments should not affect the hazard and that,
based on this fact, a vector of rank test statistics will be centered
at zero and asymptotically normal when evaluated at the true coefficient
parameter. We can then get the coefficient estimates using what Bijwaard
(2008) calls the ``inverse rank estimation'' approach of Tsiatis
(1990): using the rank test statistics as moment functions by setting
the sample versions of them to zero.\footnote{Bijwaard (2008) and Tsiatis (1990) refer to these moment functions
as \textquotedbl estimating equations\textquotedbl .}

Bijwaard (2009) points out the IVLR estimator is similar to GMM and
that there is a choice of weighting matrix. However, his IVLR estimator
is not just similar to a GMM estimator -- it is a GMM estimator.
I point this out because the ability to choose the weighting matrix
means that IVLR is actually an \emph{IV-GMM} estimator for hazard
models. In particular, it implies that the standard properties for
IV-GMM estimators carry over to his IVLR estimator. Couching Bijwaard's
(2008) IVLR estimator in the IV-GMM framework allows us to understand
it much more deeply. He specifically proposes that IVLR be implemented
as \emph{efficient IVLR}, where (an estimate of) the asymptotically
efficient weighting matrix is used -- this makes \emph{efficient
IVLR} a \emph{feasible efficient IV-GMM} estimator.

However, using feasible efficient IV-GMM entails robustness being
sacrificed for the sake of inferential efficiency: it is well-known
that there is an inherent tradeoff between efficiency and robustness
for IV-GMM and that feasible efficient IV-GMM sacrifices the latter
for the former. Therefore, I argue that an alternative instrumental
variables technique for hazard models, more robust but less efficient
than efficient IVLR, would be more practical, applicable, and attractive
for empirical studies. After all, this tradeoff is especially well
studied in linear models, where a strong preference for robustness
over efficiency has developed, as evidenced by the fact that 2SLS
is much more popular than feasible efficient IV-GMM. Thus, I propose
using the control function approach instead, based on the fact that
it satisfies these criteria -- standard results in the literature.

\subsubsection*{The Control Function Approach}

This comprehensive theory guide is the first to provide a complete
and rigorous theoretical foundation for the use of the control function
approach to instrument hazard models, undergirding existing work.
The first proposals on using the control functions in hazard models
were made by Terza et al.\ (2008), Atiyat (2011), and Palmer (2013):
Terza et al.\ (2008) and Atiyat (2011) proposed using the control
function approach to instrument continuous-time hazard models, and
Palmer (2013) proposed using the control function approach to instrument
the Prentice and Gloeckler (1978) discrete-data model. Based on Palmer
(2013), multiple applied works have used control functions to instrument
the discrete-data Prentice and Gloeckler (1978) model (Coviello et
al., 2015; Ganong \& Noel, 2022; Liebersohn \& Rothstein, 2023; Palmer,
2023). Similarly, inspired by Terza et al.\ (2008) and Atiyat (2011),
there are also a number of applied works that have used control functions
to instrument continuous-time models (e.g., Gore et al., 2010). Yet,
a critical and pressing problem is that a complete theoretical justification
is missing from the literature, as pointed out by Palmer (in personal
communication) and by others in the literature (Choi \& O'Malley,
2017).

Atiyat (2011) is a theory paper that explores the use of control functions
in (continuous-time) hazard models, but ultimately does not provide
a theoretical justification for their use. Specifically, he assumes
that control functions can be used to instrument endogenous regressors
but does not justify this or investigate when this is valid. Instead,
Atiyat (2011) points towards Terza et al.\ (2008) for justification
of this assumption

Unfortunately, Terza et al.\ (2008) provide little justification
themselves. Firstly, they make the assumption that\emph{ }\textit{\emph{the
first-stage residuals exactly equal any relevant unobserved confounders}},
which is an extremely strong assumption that cannot be justified in
empirical applications! Secondly, their theoretical discussion of
the control function approach is very brief and incomplete. In particular,
they do not derive or present critical conditions that are required
for the control function approach to be generally valid, such as what
I call the \textit{instrument conditional independence} assumption
-- the control function equivalent of the 2SLS \textit{exclusion
restriction}!\footnote{The instrument conditional independence assumption is a generic requirement.
An exception is that, when the second stage is probit, a weaker, conditional
mean restriction can be used instead (Rivers \& Vuong, 1988).} As another example, they do not prove how and under what conditions
(e.g., the Lindeberg condition) their discussion would carry over
to survival contexts.

Palmer (2013) proposed the use of the control function approach in
discrete-data hazard models and utilized this to study loan default.
However, the theoretical justification was heuristic and the task
of developing rigorous theory on instrumenting hazard models was left
as future work. This paper serves to fill in that gap in the literature.
Moreover, this paper serves as a theory guide that complements and
theoretically justifies Palmer (2023), which is an updated version
of Palmer (2013) that also contains an empirical guide as well as
a summary of the theory introduced in this paper.

This paper acts as a cornerstone for existing and future literature.
It provides a rigorous theoretical basis for and advances the theory
around the use of the control function approach to instrument hazard
models, contributing to a nascent literature and validating multiple
existing empirical studies. In addition, I also advance and broaden
the theory around the use of control functions in discrete choice
models -- after all, discrete-data hazard models can nominally be
expressed as discrete choice models (Allison, 1982; Jenkins, 1995).\footnote{I focus on the Prentice and Gloeckler (1978) discrete-data proportional
hazards model, but the methods and discussion developed here can be
easily applied more generally to other hazard models and discrete
choice models by modifying the primary model and associated \textit{control
function specification} assumption.} Existing theory papers on this mainly have a narrow focus on probit
or Poisson regression due to convenient properties that are not found
elsewhere (e.g., see Wooldridge, 2015); this paper relaxes the dependence
on these properties. I also bring formal theoretical clarification
to the Prentice and Gloeckler (1978) model by pointing out some of
the unstated assumptions, such as the Lindeberg condition.

\section{Theory\label{sec:theory}}

In this section, I introduce the version of the Prentice and Gloeckler
(1978) discrete-data proportional hazards model instrumented via the
control function approach. I discuss the Prentice and Gloeckler (1978)
model specifically because it is the most popular discrete-data proportional
hazards model and has certain advantages. Namely, it is simple, parsimonious,
and easy to implement, and it is also \emph{flexibly parametric} --
i.e., it is a distribution-free parametric model. Specifically, the
Prentice and Gloeckler model estimates the baseline hazard with a
stepped function, and this behavior is represented by the time fixed
effects.\footnote{The model is not semi-parametric, but is similar in that it is distribution-free.
On this basis, Han and Hausman (1990) also refer to flexibly parametric
estimation as \textquotedblleft nonparametric\textquotedblright{}
in a broad sense, but this usage has not been popularized.}

I then present sufficient conditions for consistent estimation and
inference as well as a consistent estimator of the asymptotic variance.
I also discuss practical considerations, such as the interpretation
and selection of the control function and the robustness/asymptotic
inefficiency of the estimator. After these is a discussion of extensions
to multiple/transformed/discrete first-stage dependent variables.

\subsection{The IV Prentice and Gloeckler (1978) Model with Control Functions\label{subsec:theory-modelintro}}

\subsubsection*{Model Overview}

$l_{1}$ and $l_{2}$ denote the log-likelihoods corresponding to
the first and second stages, respectively; their arguments are suppressed
for brevity. $y_{it}$ denotes the outcome; here, it is an indicator
for event occurrence -- a dummy variable that is 1 in the time period
a failure occurs and 0 otherwise. $x_{it}$ denotes the continuous
endogenous regressor. Let $z_{it}\coloneqq(z_{1it}',z_{2it}')$',
where $z_{1it}\coloneqq(\varphi_{t}',\mathfrak{z}_{it}')$ is a transposed
row vector of exogenous regressors and $z_{2it}$ is a transposed
row vector of ``excluded instruments'' (i.e., ``real'' instruments
that aren't the control variables $z_{1it}$). $\varphi_{t}$ represents
time-period dummies, and $\mathfrak{z}_{it}$ represents other added
control variables. As per Prentice and Gloeckler (1978), we implicitly
assume that all regressors and instruments are constant within each
time period, which is reflected in the notation.

The primary model is
\[
\begin{aligned}y_{it} & =\mathbf{1}(z_{1it}'\beta_{z}+\beta_{2}x_{it}+u_{1it}>0)\\
 & \equiv\mathbf{1}(\varphi_{t}'\psi+\mathfrak{z}_{it}'\beta_{1}+\beta_{2}x_{it}+u_{1it}>0)\equiv\mathbf{1}(\psi_{t}+\mathfrak{z}_{it}'\beta_{1}+\beta_{2}x_{it}+u_{1it}>0),
\end{aligned}
\]
where $u_{it}$ is a composite error term with some unknown distribution.
$u_{it}$ represents a combination of unobserved heterogeneity that
is potentially correlated with $x_{it}$ (which could induce omitted-variable
bias) and the ``true'' error term, which has a standard Gumbel distribution.
The primary model here is just the equivalent nominal cloglog (\emph{complementary
log-log}) representation of the Prentice and Gloeckler (1978) grouped-time
proportional hazard model, albeit in a latent variable form.\footnote{\label{fn:errordistisGLMlink}Any Bernoulli model with $\Pr(y_{1}=1|x)=F_{-\upsilon}(x'\beta)$
for some function $F_{-\upsilon}$ has a latent index representation$y_{1}=\boldsymbol{1}(y_{1}^{*}>0)$,
where $y_{1}^{*}\coloneqq x'\beta+\upsilon$ and $\upsilon$ has CDF
$F_{\upsilon}(\cdotp)$ satisfying $F_{\upsilon}(a)\equiv1-F_{-\upsilon}(-a)$.
In other words, we simply construct $\upsilon$ such that $F_{-\upsilon}(\cdotp)$
is the CDF of $-\upsilon$. This latent index representation is mechanically
identical to the original Bernoulli model because $\Pr(y_{1}=1|x)=\Pr(y_{1}^{*}>0|x)=\Pr(-\upsilon<x'\beta|x)=F_{-\upsilon}(x'\beta).$
In the cloglog model, we have $F_{-\upsilon}(a)=1-e^{-e^{a}}$, which
gives us $F_{\upsilon}(a)=e^{-e^{-a}}$, the CDF of the standard Gumbel
distribution. ($\upsilon$ and $a$ are placeholder variables.)}\textsuperscript{,}\footnote{The Prentice and Gloeckler (1978) model has an equivalent nominal
cloglog representation where we nominally treat each observation as
IID. In general, discrete-time hazard models are identical to a discrete
choice model with (1) an appropriate link function, (2) an indicator
of the event being the outcome, (3) the regressors augmented by time
fixed effects, and (4) with all observations for each entity only
being kept until ``failure'' or censoring (Allison, 1982; Jenkins,
1995). However, the observations are not truly IID: observations for
entity $i$ are cut off after ``failure''. In this sense, the model
is not a true cloglog model.} $\psi_{t}$ is a time-period fixed effect, and $(\beta_{1}',\beta_{2})'$
is the parameter vector of interest.

We can introduce the decomposition
\[
u_{1it}\eqqcolon c(v_{it})+e_{it},
\]
where we assume that the function $c(\cdot)$ is a function to be
estimated, and that $v_{it}$ is the error of the auxiliary model
\[
x_{it}=z_{it}'\pi+v_{it}\equiv z_{1it}'\pi_{1}+z_{2it}'\pi_{2}+v_{it},
\]
where $\pi\coloneqq(\pi_{1}',\pi_{2}')'$. The decomposition just
represents what happens with the new error term when the control function
$c(v_{it})$ is introduced to the second stage model. For the purposes
of this paper, we will assume that $c(v_{it})$ is a $Q$-th order
polynomial with
\[
c(v_{it})=p(v_{it})'\beta_{3},
\]
where the vector of control function terms $p(v_{it})$ is defined
by $p(v_{it})\coloneqq(v_{it},v_{it}^{2},\cdots,v_{it}^{Q})'$.\footnote{For expositional convenience, $p(v_{it})$ is defined as not containing
a $1$ because $c(v_{it})$ would otherwise contain a constant, making
it perfectly collinear with the time-period fixed effects (if all
time-period fixed effects are included).} Using the decomposition, we have
\[
y_{it}=\mathbf{1}(\psi_{t}+\mathfrak{\zeta}_{it}'\beta+e_{it}>0)\equiv\mathbf{1}(\psi_{t}+\mathfrak{z}_{it}'\beta_{1}+\beta_{2}x_{it}+p(v_{it})'\beta_{3}+e_{it}>0),
\]
where $\mathfrak{\zeta}_{it}\coloneqq(\mathfrak{z}_{it}',x_{it},p(v_{it})')'$
and $\beta\coloneqq(\beta_{1}',\beta_{2}',\beta_{3}')'$. Also, let
$\gamma\coloneqq(\psi',\beta')'.$

Under certain conditions,\footnote{The inclusion of all non-collinear second-stage controls in the first
stage is important not just for better efficiency but also robustness.
Imagine adding some exogenous regressor $\tilde{z}_{1it}$ to the
second stage but not the first stage. Since it is not in $z_{it}$,
there is no guarantee that $e_{it}|x_{it},z_{1it},\tilde{z}_{1it},p(v_{it})\sim\mathrm{Gumbel}(0,1)$
-- see Appendix \ref{sec:proof}. This is just like in 2SLS, where
including the exogenous regressor $\tilde{z}_{1it}$ in the first
stage guarantees that it will be orthogonal to $v_{it}$. (In 2SLS,
$v_{it}$ shows up in the new, composite second-stage error term.)
Obviously, we do not need to add $\tilde{z}_{1it}$ to the first stage
(which would entail dropping some instruments collinear with $\tilde{z}_{1it}$)
if $\tilde{z}_{1it}$ is perfectly collinear with the first-stage
variables $z_{1it}$: after all, $v_{it}$ would be exactly the same
either way. More formally, using the fact that $(x_{it},z_{it})\indep e_{it}|p(v_{it})$
and the perfect collinearity of $\tilde{z}_{1it}$ with $z_{it}$,
we have that $e_{it}|x_{it},z_{1it},p(v_{it})=e_{it}|x_{it},z_{it},p(v_{it})=e_{it}|x_{it},z_{it},\tilde{z}_{1it},p(v_{it})=e_{it}|x_{it},z_{1it},\tilde{z}_{1it},p(v_{it})$.
Therefore, when $\tilde{z}_{1it}$ is perfectly collinear with $z_{1it}$,
we still have that $e_{it}|x_{it},z_{1it},\tilde{z}_{1it},p(v_{it})\sim\mathrm{Gumbel}(0,1)$
even if $\tilde{z}_{1it}$ is omitted from the second stage.}
\[
e_{it}|x_{it},z_{1it},p(v_{it})\sim\mathrm{Gumbel}(0,1).
\]
This just means that the primary model is a cloglog model when the
control function terms $p(v_{it})$ are included alongside the original
second-stage regressors $x_{it}$ and $z_{1it}$. Therefore, this
implies that including the terms of $p(v_{it})$ as controls allows
us to recover an exogenous Prentice and Gloeckler (1978) model when
the base assumptions for it are satisfied. Consequently, if $p(v_{it})$
were known, we could just do that to get consistent estimates of $\beta_{1}$
and $\beta_{2}$, but we cannot because $v_{it}$ is unobserved.\footnote{Strictly speaking, just like other common binary outcome models, we
are actually estimating the scaled parameters (the ratio of the unscaled
parameters to the standard deviation of $e_{it}$) -- the unscaled
are not identified (Wooldridge, 2010). After all, $y_{it}^{*}$ has
a purely ordinal meaning, so the scale of $\beta$ and $e_{it}$ in
$y_{it}^{*}=\psi_{t}+\mathfrak{\zeta}_{it}'\beta+e_{it}$ is entirely
arbitrary. The convention for probit and logit is to set the variance
of $e_{it}$ to 1, causing $e_{it}$ have a standard normal/logistic
distribution and the scaled parameters to equal the unscaled. However,
the convention for cloglog is instead to set the variance of $e_{it}$
to $\sfrac{\pi^{2}}{6}$, which is the variance of the standard Gumbel
distribution. Technically, this means that cloglog point estimates
do not truly estimate $\gamma$, but this does not matter (except
in simulations).} Instead, if we have an estimate $\hat{\pi}$ of the auxiliary model
parameters, we can consider substituting $v_{it}$ with $\hat{v}_{it}\coloneqq x_{it}-z_{it}'\hat{\pi}$
when estimating.\footnote{Strictly speaking, $\hat{v}_{it}$ is not an estimate of $v_{it}$
because $v_{it}$ is not a parameter.} To see why, first substitute the unobserved variable $v_{it}$ for
the observed variables in $z_{it}$ and unknown parameter vector $\pi$
using the auxiliary model. This gives us
\[
y_{it}=\mathbf{1}(\psi_{t}+\mathfrak{z}_{it}'\beta_{1}+\beta_{2}x_{it}+p(x_{it}-z_{it}'\pi)'\beta_{3}+e_{it}>0).
\]
This is a classic two-step Z-estimation setup: were it not for $\pi$
being unknown, we would be able to directly estimate $\psi_{t}$ and
$\beta$ using MLE.\textbf{}\footnote{\textit{M-estimation theory} refers to extremum estimation theory
in general, whereas \textit{Z-estimation theory} (coined by Van der
Vaart, 1998) is a special case where first-order condition optimization
is possible. Classical GMM, such as Hansen (1982) or Newey and McFadden
(1994), covers ergodic or IID data, and its theorems are special cases
of finite-dimensional Z-theorems. \textit{Non-classical GMM} simply
refers to finite-dimensional Z-theorems that use different CLTs to
achieve analogous results. One famous example of a finite-dimensional
Z-theorem is Theorem 3.3 in Pakes and Pollard (1989), which is a slightly
more general version of Theorem 3 in Huber (1967) that allows the
sample moment functions to be non-smooth and even discontinuous. It
forms the theoretical basis of the \textit{Simulated Method of Moments}
(SMM), also known as the \textit{Method of Simulated Moments} (MSM),
where moment functions can be simulated to avoid the intractability
or non-existence of an analytical form.} Under appropriate assumptions, we can show that we can replace $\pi$
with $\hat{\pi}$ and instead do \emph{QMLE} -- quasi-maximum likelihood
estimation.\footnote{Technically, this is specifically \emph{2SCMLE} -- two-stage conditional
maximum likelihood estimation (Rivers \& Vuong, 1988).}

Notice that, just like in sequential classical GMM (Newey, 1984; Newey
\& McFadden, 1994), the sequential nature of this estimator means
that we can ``stack'' the moment functions of the first and second
stages (namely, the first-stage score and second-stage quasi-score).
This is because the just-identification of both stages implies that
the vector of moment conditions in each can be solved for exactly,
implying that combining them into one vector of moment conditions
yields the exact same estimator.\footnote{``Just-identification'' is used here in the general sense (one moment
condition for each parameter) rather than the more common meaning
for IV (the numbers of ``excluded instruments'' and endogenous regressors
being equal).} In other words, we can treat the estimator like a one-step Z-estimator,
allowing us to apply a Z-theorem to find sufficient conditions and
get results.

\subsubsection*{Moment Functions}

Let $\theta\coloneqq(\pi',\psi',\beta')'$. The moment functions are
\[
\hat{m}(\theta)\coloneqq\left(\begin{array}{c}
\hat{m}_{1}(\theta)\\
\hat{m}_{2}(\theta)
\end{array}\right)\equiv\left(\begin{array}{c}
\frac{1}{n}\stackrel[i=1]{n}{\sum}g_{1i}(\theta)\\
\frac{1}{n}\stackrel[i=1]{n}{\sum}g_{2i}(\theta)
\end{array}\right)\equiv\frac{1}{n}\stackrel[i=1]{n}{\sum}g_{i}(\theta),
\]
where
\[
\hat{m}_{1}(\theta)\equiv\hat{m}_{1}(\pi)\coloneqq\frac{1}{n}\frac{\partial l_{1}}{\partial\pi}\equiv\frac{1}{n}\stackrel[i=1]{n}{\sum}\stackrel[t=1]{T}{\sum}z_{it}'v_{it}\equiv\frac{1}{n}\stackrel[i=1]{n}{\sum}\underset{g_{1i}(\theta)}{\underbrace{\stackrel[t=1]{T}{\sum}z_{it}'(x_{it}-z_{it}'\pi)}}
\]
and\footnote{Strictly speaking, Prentice and Gloeckler (1978) actually use the
unnormalized score (i.e., $n\times\hat{m}(\theta)$) in their paper.
I normalize by $n$ here to emphasize the stochastic convergence of
the score.}
\[
\hat{m}_{2}(\theta)\equiv\hat{m}_{2}(\psi,\beta,\pi)\coloneqq\frac{1}{n}\left(\begin{array}{c}
\frac{\partial l_{2}}{\partial\psi_{1}}\\
\vdots\\
\frac{\partial l_{2}}{\partial\psi_{T}}\\
\frac{\partial l_{2}}{\partial\beta_{1}}\\
\vdots\\
\frac{\partial l_{2}}{\partial\beta_{K}}
\end{array}\right)\equiv\frac{1}{n}\left(\begin{array}{c}
\frac{\partial}{\partial\psi_{1}}\stackrel[i=1]{n}{\sum}l_{2i}\\
\vdots\\
\frac{\partial}{\partial\psi_{T}}\stackrel[i=1]{n}{\sum}l_{2i}\\
\frac{\partial}{\partial\beta_{1}}\stackrel[i=1]{n}{\sum}l_{2i}\\
\vdots\\
\frac{\partial}{\partial\beta_{K}}\stackrel[i=1]{n}{\sum}l_{2i}
\end{array}\right)\equiv\frac{1}{n}\stackrel[i=1]{n}{\sum}\underset{\mathclap{\coloneqq g_{2i}(\theta)\equiv\frac{1}{T}\stackrel[t=1]{T}{\sum}g_{2it}(\theta)}}{\underbrace{\left(\begin{array}{c}
\frac{\partial l_{2i}}{\partial\psi_{1}}\\
\vdots\\
\frac{\partial l_{2i}}{\partial\psi_{T}}\\
\frac{\partial l_{2i}}{\partial\beta_{1}}\\
\vdots\\
\frac{\partial l_{2i}}{\partial\beta_{K}}
\end{array}\right)}},
\]
where $n$ is the number of entities (indexed by $i$), $T$ is the
number of time periods, and $K$ is the number of regressors. The
elements of $g_{2i}(\theta)$ are defined by
\[
\frac{\partial l_{2i}}{\partial\psi_{t}}=\begin{cases}
-h_{it} & t<s\\
\delta_{i}b_{it} & t=s\\
0 & t>s
\end{cases};\qquad\frac{\partial l_{2i}}{\partial\beta_{k}}=\mathfrak{\hat{\zeta}}_{is,k}\left(\delta_{i}b_{is}-\stackrel[j=1]{s-1}{\sum}h_{ij}\right),
\]
where $\delta_{i}$ is the censoring indicator, $t$ represents a
time period, $s$ is the time period of survival (with subscript $i$
suppressed for clarity), $\mathfrak{\hat{\zeta}}_{it,k}$ is the $k$-th
element of $\mathfrak{\hat{\zeta}}_{it}$ (i.e., the observation for
the $k$-th regressor, excluding the time fixed effects, for entity
$i$ and time $t$$)$, and where
\[
b_{is}\coloneqq\frac{h_{it}e^{-h_{it}}}{1-e^{-h_{it}}};\qquad h_{it}\coloneqq e^{\psi_{t}+\mathfrak{\zeta}_{it}'\beta}.
\]
$\hat{m}_{1}(\theta)$ is a score function, but $\hat{m}_{2}(\theta)$
is technically a quasi-score function because $p(\hat{v}_{it})$ is
a vector of generated regressors. $g_{1i}(\theta)$ is the entity-specific
first-stage score contribution,$g_{2i}(\theta)$ is the entity-specific
second-stage quasi-score contribution, and $g_{2it}(\theta)$ is the
observation-specific second-stage quasi-score contribution; $g_{i}(\theta)\coloneqq(g_{1i}(\theta)',g_{2i}(\theta)')'$.
$\hat{\theta}\coloneqq(\hat{\pi}',\hat{\psi}',\hat{\beta}')'$ is
the solution of
\[
\hat{m}(\hat{\theta})=0.
\]

\subsection{Root-$n$ Asymptotic Theory\label{subsec:theory-asymptotics}}

\textbf{}Here, I present assumptions that are sufficient to infer
that standard GMM asymptotic properties hold, such as consistency
and asymptotic normality for $\hat{\theta}$.
\begin{assumption}[Basic Setup]
\label{assu:setup}\hfill
\begin{itemize}
\item Regressors and instruments are constant within time period.
\item Strict exogeneity.
\item Censoring and failure are independent.
\item Observations are independent across individuals.
\item The true parameter is finite-dimensional.
\end{itemize}
\end{assumption}
\begin{rem}
Strict exogeneity is just a standard, fundamental assumption for hazard
models to ensure identifiability. It simply requires that only the
contemporaneous values of variables matter, and this has already been
assumed, given that it is an inherent part of the structure of the
model introduced in Section \ref{subsec:theory-modelintro}.
\end{rem}
\begin{assumption}[Full Rank Regressors]
\label{assu:rank}\hfill
\begin{itemize}
\item \textbf{Instrument Rank Condition:} $\mathrm{rank}(\mathbb{E}\left[z_{it}z_{it}'\right])=L_{Z}$,
where $L_{Z}$ is the (fixed) number of instruments, including the
exogenous regressors.
\item \textbf{Primary Model Rank Condition:} $\mathrm{rank}(\mathbb{E}\left[\xi_{it}\xi_{it}'\right])=L_{\Xi}$,
where $L_{\Xi}$ is the (fixed) number of augmented second-stage regressors
(i.e., including the control function terms). Note that two simple
necessary conditions (which in particular are necessary for $v_{it}$
to not be perfectly collinear with other regressors) are:
\begin{itemize}
\item Instrument Relevance Condition: $\pi_{2}\neq0$. (At least one of
the ``excluded instruments'' must be relevant.)
\item Instrument Order Condition: $L_{Z}\geq L_{X}$, where $L_{X}$ is
the (fixed) number of original second-stage regressors. (The number
of instruments must equal or exceed the number of original regressors,
i.e., there must be at least one ``excluded instrument''.)
\end{itemize}
\end{itemize}
\end{assumption}
\begin{rem}
The instrument and primary model rank conditions are equivalent to
saying that the expected first-stage and second-stage Hessians, evaluated
at the true parameter, are invertible. In particular, the instrument
rank condition is required for $\pi_{2}$, the first-stage coefficients
on the excluded instruments, to be identified -- it embodies the
IV relevance condition. In addition, given the instrument rank condition,
the primary model rank condition guarantees the identification of
the second-stage coefficients $\gamma$. (Note that both conditions
in particular require the number of observations to be at least (1)
the number of first-stage regressors and (2) the number of second-stage
regressors.) Assumption \ref{assu:rank} does not imply that the estimator
always exists but rather that the probability of it existing converges
to one as the sample size increases.
\end{rem}
\begin{assumption}[Instrument Conditional Independence]
\label{assu:condindep}$u_{it}$ and $z_{it}$ are independent conditional
on $v_{it}$. A sufficient condition is that $(u_{it},v_{it})$ and
$z_{it}$ are independent. (Note that this sufficient condition requires
that $x_{it}$ be a continuous variable.)
\end{assumption}
\begin{rem}
Instrument conditional independence is the control function equivalent
of the 2SLS \textit{\emph{exclusion restriction.}}
\end{rem}
\begin{assumption}[Control Function Specification]
\label{assu:CFspec}$u_{it}=c(v_{it})+e_{it}$ with $e_{it}|v_{it}\sim\mathrm{Gumbel}(0,1)$,
where we will assume for the purposes of this paper that the control
function $c(v_{it})$ is a $Q$-th order polynomial with $c(v_{it})\coloneqq p(v_{it})'\beta_{3}$,
where the vector of control function terms $p(v_{it})$ is defined
by $p(v_{it})\coloneqq(v_{it},v_{it}^{2},\cdots,v_{it}^{Q})'$.
\end{assumption}
\begin{rem}
Note that $e_{it}|p(v_{it})\equiv e_{it}|v_{it}$, and therefore $e_{it}|p(v_{it})\sim\mathrm{Gumbel}(0,1)$,
because $p(v_{it})$ depends only on $v_{it}$ and contains $v_{it}$
as one of its elements. Thus, this assumption just means that $c(\cdot)$
has a sufficiently flexible function form so as to make $e_{it}$
have a standard Gumbel distribution conditional on $p(v_{it})$. (That
this holds perfectly is typically not true in practice, but the point
is that we want a sufficiently good approximation.)
\end{rem}
\begin{rem}
\label{rem:distlinkhazardequiv}The assumption that the $e_{it}$,
the error term in the latent variable representation of $y_{it}$,
has a standard Gumbel distribution (conditional on $v_{it}$) is exactly
equivalent to assuming that the second-stage model is cloglog.\footnote{Recall that the assumption of the (conditional) distribution of the
error term in the latent variable representation of any binary choice
model is exactly equivalent to the assumption of the GLM link function
for that binary choice model. For more details, see Footnote \ref{fn:errordistisGLMlink}
in Section \ref{subsec:theory-modelintro}.} Therefore, given that the Prentice and Gloeckler (1978) discrete-data
proportional hazards model has an equivalent nominal cloglog representation,
the combination of Assumptions \ref{assu:condindep} and \ref{assu:CFspec}
is essentially a proportional hazards assumption. More specifically,
they imply that we have proportional hazards after including the control
function terms as control variables.
\end{rem}
\begin{assumption}[Regularity Conditions]
\label{assu:regularity}\hfill
\begin{itemize}
\item \textbf{Lindeberg Condition:} The Lindeberg condition holds for the
individual-specific second-stage quasi-score contributions. Using
notation introduced later, this means the elements of $g_{2i}(\theta)$,
where $g_{2i,j}(\theta)$ represents the $j$-th element, have finite
mean and variance and satisfy for all $\epsilon>0$: 
\[
\underset{n\rightarrow\infty}{\lim}\frac{1}{s_{n}^{2}}\mathbb{E}\left[\mathrm{Var}(g_{2i,j}(\theta))\cdot\boldsymbol{1}\left(\left|g_{2i,j}(\theta)-\mathbb{E}\left[g_{2i,j}(\theta)\right]\right|>\epsilon s_{n}\right)\right]=0,
\]
where $s_{n}^{2}\coloneqq\stackrel[j=1]{n}{\sum}\mathrm{Var}(g_{2i,j}(\theta))$.
\item \textbf{Parameter Space Compactness:} The true coefficient parameter
must be in the interior of a compact parameter space. A sufficient
condition is that it is a real-valued vector.
\end{itemize}
\end{assumption}
\begin{thm}[Consistency and Asymptotic Normality of Control Function Hazard Model
Estimator]
\label{thm:CAN}Under Assumptions \ref{assu:setup}, \ref{assu:rank},
\ref{assu:condindep}, \ref{assu:CFspec}, and \ref{assu:regularity},
the (quasi-)maximum likelihood estimates $\hat{\theta}$ of $\theta$
will be consistent and asymptotically normal with known variance $V\coloneqq G^{-1}\Omega(G^{-1})'$,
where
\[
\sqrt{n}(\hat{\theta}-\theta_{0})\xrightarrow{d}N(0,G^{-1}\Omega(G^{-1})'),
\]
with\footnote{$\nabla_{\mkern-4mu \theta}$ is the gradient w.r.t.\ $\theta$,
but is defined according to the standard convention in economics:
that the derivative w.r.t.\ $\theta$ has the same dimensions as
$\theta$ and that the gradient w.r.t.\ $\theta$ has transposed
dimensions (which is the opposite of the standard convention in mathematics).}
\[
G\coloneqq\mathbb{E}\left[\nabla_{\mkern-4mu \theta}\,g_{i}(\theta_{0})\right];\qquad\Omega\coloneqq\mathbb{E}\left[g_{i}(\theta_{0})g_{i}(\theta_{0})'\right].
\]
A consistent estimator of $V\equiv G^{-1}\Omega(G^{-1})'$ is $\widehat{V}\coloneqq\widehat{G}^{-1}\widehat{\Omega}(\widehat{G}^{-1})'$,
where $\widehat{\Omega}$ can be any consistent estimator of $\Omega$
and $\widehat{G}$ is the consistent estimator of $G$ defined by
$\widehat{G}\coloneqq\frac{1}{n}\sum_{i=1}^{n}\negthickspace\nabla_{\mkern-4mu \theta}\,g_{i}(\theta_{0})$.
\end{thm}
$G$ has the structure
\begin{align*}
G & \equiv\left[\begin{array}{c|c}
G_{11} & G_{12}\\
\hline G_{21} & G_{22}
\end{array}\right]\\
 & \equiv\left[\begin{array}{c|c}
\mathbb{E}\left[-z_{it}z_{it}'\right] & \text{\raisebox{-0.2em}{\LARGE\ensuremath{\boldsymbol{0}}}}\\
\hline \mathbb{E}\left[\mkern3mu\mathchar'26\mkern-12mu d_{it}\xi_{it}z_{it}'\right] & \mathbb{E}\left[d_{it}\xi_{it}\xi_{it}'\right]
\end{array}\right],
\end{align*}
where $\xi_{it}\coloneqq(\varphi_{t}',\mathfrak{\zeta}_{it}')'\equiv(\varphi_{t}',\mathfrak{z}_{it}',x_{it},p(\hat{v}_{it})')'$
is a transposed row vector of second-stage regressors, $z_{it}$ is
a transposed row vector of auxiliary model instruments (i.e., the
same first-stage regressors as before), $\gamma$ is the coefficient
vector on $\xi_{it}$, and $\mkern3mu\mathchar'26\mkern-12mu d_{it}$
and $d_{it}$ are defined by
\[
\mkern3mu\mathchar'26\mkern-12mu d_{it}=-d_{it}\stackrel[q=1]{Q}{\sum}\beta_{3q}v_{it}^{q-1};\qquad d_{it}=\begin{cases}
e^{\xi_{it}'\gamma}\frac{e^{e^{\xi_{it}'\gamma}}-e^{\xi_{it}'\gamma+e^{\xi_{it}'\gamma}}-1}{\left(e^{e^{\xi_{it}'\gamma}}-1\right)^{2}} & y_{it}=1\\
-e^{\xi_{it}'\gamma} & y_{it}=0
\end{cases},
\]
with $Q$ being the order of the control function polynomial and $\beta_{3q}$
being the coefficient on $v^{q}$.

$\widehat{G}$ has the structure

\begin{align*}
\widehat{G} & \equiv\left[\begin{array}{c|c}
\widehat{G}_{11} & \widehat{G}_{12}\\
\hline \widehat{G}_{21} & \widehat{G}_{22}
\end{array}\right]\\
 & \equiv\left[\begin{array}{c|c}
\frac{1}{\tilde{n}}\stackrel[i=1]{n}{\sum}\stackrel[t=1]{T}{\sum}-z_{it}z_{it}' & \text{\raisebox{-0.2em}{\LARGE\ensuremath{\boldsymbol{0}}}}\\
\hline \frac{1}{\tilde{n}}\stackrel[i=1]{n}{\sum}\stackrel[t=1]{T}{\sum}\mathrlap{\hskip1ex\hat{\vphantom{d}}}\mkern3mu\mathchar'26\mkern-12mu d_{it}\hat{\xi}_{it}z_{it}' & \frac{1}{\tilde{n}}\stackrel[i=1]{n}{\sum}\stackrel[t=1]{T}{\sum}\hat{d}_{it}\hat{\xi}_{it}\hat{\xi}_{it}'
\end{array}\right]\\
 & \equiv\left[\begin{array}{c|c}
-Z'Z & \text{\raisebox{-0.2em}{\LARGE\ensuremath{\boldsymbol{0}}}}\\
\hline \widehat{\Xi}'\mathrlap{\hskip1.1ex\widehat{\vphantom{D}}}\mkern2mu\rule[0.75ex]{0.75ex}{0.06ex}\mkern-8mu DZ & \widehat{\Xi}'\widehat{D}\widehat{\Xi}
\end{array}\right].
\end{align*}
$\tilde{n}$ is the total number of observations -- i.e., the number
of pairs $(i,t)$. $\widehat{\Xi}$ and $Z$ are the matrix versions
of $\hat{\xi}_{it}\coloneqq(\varphi_{t}',\mathfrak{\hat{\zeta}}_{it}')'$
and $z_{it}$ -- i.e., $\hat{\xi}_{it}'$ and $z_{it}'$ stacked;
$\hat{\xi}_{it}$ only differs from $\xi_{it}$ in that it contains
$p(\hat{v}_{it})$ rather than $p(v_{it})$. $\mathrlap{\hskip1.1ex\widehat{\vphantom{D}}}\mkern2mu\rule[0.75ex]{0.75ex}{0.06ex}\mkern-8mu D$
and $\widehat{D}$ are defined by
\[
\mathrlap{\hskip1.1ex\widehat{\vphantom{D}}}\mkern2mu\rule[0.75ex]{0.75ex}{0.06ex}\mkern-8mu D\coloneqq\mathrm{diag}\left(\mathrlap{\hskip1ex\hat{\vphantom{d}}}\mkern3mu\mathchar'26\mkern-12mu d_{it}\right);\qquad\widehat{D}\coloneqq\mathrm{diag}\left(\hat{d}_{it}\right),
\]
with
\[
\mathrlap{\hskip1ex\hat{\vphantom{d}}}\mkern3mu\mathchar'26\mkern-12mu d_{it}=-\hat{d}_{it}\stackrel[q=1]{Q}{\sum}\hat{\beta}_{3q}\hat{v}_{it}^{q-1};\qquad\hat{d}_{it}=\begin{cases}
e^{\hat{\xi}_{it}'\hat{\gamma}}\frac{e^{e^{\hat{\xi}_{it}'\hat{\gamma}}}-e^{\hat{\xi}_{it}'\hat{\gamma}+e^{\hat{\xi}_{it}'\hat{\gamma}}}-1}{\left(e^{e^{\hat{\xi}_{it}'\hat{\gamma}}}-1\right)^{2}} & y_{it}=1\\
-e^{\hat{\xi}_{it}'\hat{\gamma}} & y_{it}=0
\end{cases}.
\]

The invertibility of $\widehat{G}$ is implied by the invertibility
of $Z'Z$ and $\widehat{\Xi}'\widehat{\Xi}$.
\begin{proof}
See Appendix \ref{sec:proof}.
\end{proof}
\begin{rem}
Because the asymptotics are over $i$, we could use the Lindeberg-L\'evy
CLT if the second-stage quasi-score contributions are assumed to be
IID draws (e.g., if the sampling of entities was uniformly random)
from some finite-variance distribution. This is unsurprising since
the hazard functions for each entity $i$ are modeled as a deterministic
function of only the regressors for entity $i$. However, following
Prentice and Gloeckler (1978), I instead use the Lindeberg-Feller
CLT. The Lindeberg condition is weaker than the IID assumption for
Lindeberg-L\'evy CLT, and using the Lindeberg condition allows us
to avoid making the strong assumption that the data must come from
a representative sample. Moreover, the interpretation of the Lindeberg
condition is more elucidating than the assumption of some finite-variance
distribution for all entities.

The Lindeberg condition simply means that a single entity cannot affect
the second-stage quasi-score too much. Sufficient conditions are the
regressors being bounded and the hazard rate being bounded away from
zero.\footnote{Intuitively, if the hazard rate is bounded away from zero, then the
``worst case'' is this bound. For convenience, let this bound be
a constant. The survival time at this bound follows a geometric distribution.
The law of total variance and the fact that the observation-specific
quasi-score contributions are expectation zero means that the total
variance of the entity-specific score contributions is the sum of
the variance of the entity-specific score given survival up to a certain
time period multiplied by the probability of surviving up to that
time period). The boundedness of regressors implies a linear bound
on the former, so the the latter decays faster than the former grows.
Thus the ratio test gives us convergence of the total variance, implying
that the Lindeberg condition is satisfied.} The Lindeberg condition can be thought of as a requirement that the
survival function must, on average, decay fast enough; this is because
each period adds another term to the entity-specific score contribution,
increasing its variance. Intuitively, this means that we want large-n,
small-T asymptotics. (Note that these are not ``fixed-T'' asymptotics
since the number of time periods per entity is random.) This is unsurprising
since the Prentice and Gloeckler (1978) model has an equivalent representation
as a cloglog model with time fixed effects -- this is what we would
expect from a fixed effects model.
\end{rem}
\begin{rem}
The assumption that the (finite-dimensional and fixed) true parameter
is in the interior of a compact parameter space may seem strong, but
this is actually a very weak assumption because the parameter space
is not fixed. Sufficient conditions are that the (finite-dimensional
and fixed) true parameter here is real-valued and that there are no
shape restrictions (i.e., inequality constraints on the parameters)
because then we can always choose a large enough compact parameter
space that includes said parameter in its interior. (This argument
can essentially be expressed in more general, concise, and formal
terms as the mathematical technique of \textit{compactification}.
Briefly, if you have a topological space that is non-compact, often
you can recast it -- i.e, \textit{embed} it -- as a dense subset
of a larger, compact space and use that larger, compact space instead.
See Bahadur, 1971, for a discussion of this.) Note that these sufficient
conditions preclude the inclusion of perfect predictors of the outcome
variable in the second stage because they will have infinite coefficients.

Relatedly, assuming that the true parameter is in the interior of
its parameter space is necessary for the asymptotic normality of the
corresponding estimator. When the parameter space is fixed, what we
are really saying is that we are treating the parameter space as a
search space in which we conduct constrained optimization. Why this
would be problematic is intuitive and immediately apparent: the PDF
(if it exists) must be zero outside the parameter space in the limit
because the parameter estimate can never be outside this search space.
If the true parameter is on a boundary of the parameter space, then
the asymptotic distribution of the estimate clearly cannot be normal.
Thankfully, this is not a concern here because we can always take
a compact parameter space large enough to make the true parameter
be in the interior under the aforementioned sufficient conditions.
\end{rem}

\subsubsection{The Interpretation and Selection of the Control Function\label{subsubsec:theory-CFinterpretation}}

Let $\hat{c}(\hat{v}_{it})$ denote the estimated control function,
which is composed of the first-stage residuals $\hat{v}_{it}$ fed
into some estimate $\hat{c}(\cdot)$ of $c(\cdot)$. In this paper,
I have focused on $\hat{c}(\hat{v}_{it})\coloneqq\hat{p}(v_{it})'\hat{\beta}_{3}$,
but the discussion here still applies to non-polynomial control functions.
Note that the unobserved control function $c(v_{it})$ represents
the unobserved heterogeneity in the ``true model''. We thus are
assuming that the unobserved heterogeneity takes the functional form
of $c(\cdot)$ and that it is only a one-dimensional function of $v_{it}$.

The general interpretation of the estimated control function $\hat{c}(\hat{v}_{it})$
is that it is a good proxy for the unobserved heterogeneity $c(v_{it})$.
In models where the second stage is linear or probit, a conditional
mean restriction is required rather than a conditional independence
assumption: consequently, $c(v_{it})$ just needs to flexible enough
to have the same mean as the unobserved heterogeneity, conditional
on $v_{it}$ (e.g., see Rivers and Vuong, 1988). Then, the estimated
control function $\hat{c}(\hat{v}_{it})$ is estimating the mean of
the unobserved heterogeneity. 

However, that specific interpretation is not true for non-probit (nominally)
discrete choice models because a conditional independence assumption
is required. Instead, $c(v_{it})$ must have the same distribution
as the unobserved heterogeneity; hence, a correct interpretation of
$\hat{c}(\hat{v}_{it})$ is that it is an optimal approximation in
distribution to the unobserved heterogeneity.\footnote{The control function approach can also be understood as the flexibly
parametric counterpart of the Blundell and Powell (2003, 2004) semi-parametric
approach, which does not parametrically restrict the relationship
between $u_{it}$ and $e_{it}$. (See Imbens \& Newey, 2009, for the
generalized framework and Wooldridge, 2015, for a summary.)} After all, when we carry out MLE-type estimation, what we are really
doing is fitting the (GLM) pseudo-residuals $\hat{e}_{it}$ such that
their empirical distribution is an optimal fit to the assumed distribution
of $e_{it}$, with the norm being the Kullback-Leibler divergence.\footnote{Equivalently, with the estimated control function included in the
fitted model, we attempt to make the hazard function as close to proportional
as possible in some metric. See \ref{rem:distlinkhazardequiv}.} By the additivity property of the Kullback-Leibler divergence and
the fact that the distribution of $e_{it}$ does not depend on $c(v_{it})$,
we see that $\hat{c}(\cdot)$ is fitted such that $\hat{c}(\hat{v}_{it})$
is an optimal approximation in distribution to $c(v_{it})$ in the
Kullback-Leibler sense, subject to the functional form restrictions
of $c(\cdot)$.

A key takeaway is that more flexibility in the functional form of
$c(\cdot)$ can increase improve accuracy by allowing $\hat{c}(\hat{v}_{it})$
to better approximate the distribution of the unobserved heterogeneity.
Though, too flexible a functional form can result in concerns of overfitting
and the curse of dimensionality; thus, the practitioner is looking
for a sweet spot in the middle. This is essentially a machine learning-style
dictionary selection problem, and there are many techniques from the
machine learning literature that can be used to help with this selection
process, e.g., cross-validation.

Finally, the other key takeaway is that this implies that an appropriate
robustness check is to examine how $\hat{\theta}$ changes in response
to different functional forms being specified for $c(\cdot)$, e.g.,
by varying the degree of a polynomial control function. After the
selection process, if $\hat{\theta}$ is not sensitive to slight changes
to the functional form, then this suggests that the functional form
of $c(\cdot)$ is near-optimal.

\subsubsection{Asymptotic Inefficiency\label{subsubsec:theory-effandCIME}}

\noindent Control function methods are asymptotically inefficient
(this includes 2SLS) -- compared with feasible efficient IV-GMM methods,
they trade efficiency for robustness. The parallels of QMLE to true
MLE can be leveraged to give an intuitive and in-depth explanation
of the cause of this inefficiency.

Rather than a joint likelihood, we use a conditional likelihood. In
other words, we see that the conditional information matrix does not
reflect all of the information available: we pass only the first-stage
estimates into the second stage, and any other information about the
excluded instruments is lost from the perspective of the second stage.
This intuitively explains why the second-stage asymptotic variance
matrix would be inflated.

I give a more in-depth explanation here. We effectively conduct constrained
optimization of the true likelihood when we do QMLE. Note that adding
more information tends to sharpen (i.e., increase the magnitude of
the curvature) of the objective function; thus, we are constrained
in not using that discarded information to sharpen the objective function.
Therefore, the magnitude of the curvature of the objective function
is limited. Now, note that $G_{22}$ is the expected gradient of the
second-stage quasi-score -- the negative of the expected Hessian
of the second-stage conditional log-likelihood (i.e., second-stage
quasi-log-likelihood), and thus embodies the negative of the curvature.
The limiting of the curvature therefore limits the size of the determinant
of $G_{22}$ since the determinant of a Hessian, evaluated at a critical
point of a function, equals the Gaussian curvature at that point.
Because the conditional information matrix equality implies that the
second-stage asymptotic variance equals $G_{22}^{-1}$, using the
fact that $det(A^{-1})=\sfrac{1}{det(A)}$ for an arbitrary matrix
$A$, we can see that this ultimately results in the (root-$n$-scaled)
asymptotic variance matrix being inflated in the sense that its determinant
is limited in how close to zero it can get.

\subsection{Extensions\label{subsec:theory-extensions}}

\subsubsection{Multiple Endogenous Regressors\label{subsubsec:theory-multiauxdep}}

Under very similar assumptions, it is trivial to extend Theorem \ref{thm:CAN}
to allow for multiple endogenous regressors (strictly speaking, multiple
first-stage dependent variables), i.e., when $x_{it}$ is a vector
and there are multiple first-stage regressions -- one for each element
of $x_{it}$. The only noteworthy assumption here is Assumption \ref{assu:CFspec}.
The overall, multivariate control function $c(v_{it})$ is allowed
to simply be composed of the separate, univariate control functions
for each first-stage regression, i.e., just having separate control
functions for each each endogenous regressor. However, note that $c(v_{it})$
is also allowed to be a non-separable function of all first-stage
residual terms -- e.g., a multivariate polynomial with all residual
terms as arguments. Thus, when there are multiple endogenous regressors,
there is additional flexibility in the choice of functional form for
the control function $c(v_{it})$.

The next thing to note is what the asymptotic variance $V\equiv G^{-1}\Omega(G^{-1})'$
and its estimator $\widehat{V}\equiv\widehat{G}^{-1}\widehat{\Omega}(\widehat{G}^{-1})'$
become. $\Omega$ and $\widehat{\Omega}$ have the same structure
as before, but $G$ and $\widehat{G}$ do not. Denoting the number
of first-stage dependent variables by $\kappa$, $G$ is given by
\begin{align*}
G & \equiv\left[\begin{array}{c|c}
G_{11} & G_{12}\\
\hline G_{21} & G_{22}
\end{array}\right]\\
 & \equiv\left[\begin{array}{c|c}
\mathbb{E}\left[-z_{it}z_{it}'\right] & \text{\raisebox{-0.2em}{\LARGE\ensuremath{\boldsymbol{0}}}}\\
\hline \mathbb{E}\left[\mkern3mu\mathchar'26\mkern-12mu d_{it}\xi_{it}z_{it}'\right] & \mathbb{E}\left[d_{it}\xi_{it}\xi_{it}'\right]
\end{array}\right].
\end{align*}

On the other hand, $\widehat{G}$ is given by
\begin{align*}
\widehat{G} & \equiv\left[\begin{array}{c|c}
\widehat{G}_{11} & \widehat{G}_{12}\\
\hline \widehat{G}_{21} & \widehat{G}_{22}
\end{array}\right]\\
 & \equiv\left[\begin{array}{ccc|c}
-Z^{(1)}{}'Z^{(1)} &  & \vcenter{\hbox{\text{\kern-1em  \raisebox{-1em}{\huge\ensuremath{\boldsymbol{0}}}}}}\\
 & \mathclap{\ddots} &  & \vcenter{\hbox{\text{\raisebox{0.4em}{\huge\ensuremath{\boldsymbol{0}}}}}}\\
\vcenter{\hbox{\text{\kern1em  \raisebox{1em}{\huge\ensuremath{\boldsymbol{0}}}}}} &  & -Z^{(\kappa)}{}'Z^{(\kappa)}\\
\hline \widehat{\Xi}'\mathrlap{\hskip1.1ex\widehat{\vphantom{D}}}\mkern2mu\rule[0.75ex]{0.75ex}{0.06ex}\mkern-8mu DZ^{(1)} & \mathclap{\cdots} & \widehat{\Xi}'\mathrlap{\hskip1.1ex\widehat{\vphantom{D}}}\mkern2mu\rule[0.75ex]{0.75ex}{0.06ex}\mkern-8mu DZ^{(\kappa)} & \widehat{\Xi}'\widehat{D}\widehat{\Xi}
\end{array}\right],
\end{align*}
where $Z^{(1)}$ is the matrix of instruments for the first auxiliary
equation (i.e., for the first endogenous regressor) and $Z^{(\kappa)}$
is the matrix of instruments for the last auxiliary equation (i.e.,
for the last endogenous regressor).

When identical instruments are used for all first-stage regressions,
the expression simplifies to
\[
\left[\begin{array}{c|c}
I_{\kappa}\otimes-Z'Z & \text{\raisebox{-0.2em}{\LARGE\ensuremath{\boldsymbol{0}}}}\\
\hline \boldsymbol{1}_{\kappa}'\otimes\widehat{\Xi}'\mathrlap{\hskip1.1ex\widehat{\vphantom{D}}}\mkern2mu\rule[0.75ex]{0.75ex}{0.06ex}\mkern-8mu DZ & \widehat{\Xi}'\widehat{D}\widehat{\Xi}
\end{array}\right],
\]
where $I_{\kappa}$ is an identity matrix, $\boldsymbol{1}_{\kappa}'$
is a row vector of ones, and $\otimes$ represents the Kronecker product.

\subsubsection{Discrete Endogenous Regressors That Are Functions of Continuous Variables\label{subsubsec:theory-funcauxdep}}

In the primary model, the endogenous regressors $x_{it}$ can actually
be replaced by $f(x_{it})$ for any (measurable, potentially vector-valued)
function $f$ without needing to change the auxiliary model or the
control function: Theorem \ref{thm:CAN} still carries over.\footnote{This function of $x_{it}$ being measurable just means that it is
a random variable itself.} This property is characteristic of the control function approach
and applies more generally to other models (e.g., linear or probit).
Ultimately, this is because $(x_{it},z_{it})\indep e_{it}|p(v_{it})\implies(f(x_{it}),z_{it})\indep e_{it}|p(v_{it})$:
it guarantees that $f(x_{it})$ will also be conditionally independent
of the new error term when the control function is added as a regressor,
just like $x_{it}$. In other words, when your endogenous regressors
are $f(x_{it})$, you only need to run one first-stage regression
for each element of $x_{it}$.

As an example, suppose that your model has five endogenous regressors
of the form $x_{1it}$, $x_{1it}^{2}$, $x_{2it}$, $x_{2it}^{2}$,
and $x_{1it}x_{2it}$, for some observed scalar, continuous variables
$x_{1it}$ and $x_{2it}.$ You only need to run two first-stage regressions
-- one for $x_{1it}$ and one for $x_{2it}$ -- and include the
resulting two control functions (or one non-separable, multivariate
control function).

Note that the untransformed $x_{it}$ is allowed to be absent from
the primary model. This can particularly useful -- for example, suppose
that your model has one endogenous regressor, $f(x_{it})=\boldsymbol{1}(x_{it}>0)$,
for some observed scalar, continuous variable $x_{it}$. If we want
to directly instrument $f(x_{it})$, since $\boldsymbol{1}(x_{it}>0)$
is a binary variable, the first stage is not allowed to be linear
(see below). We could estimate a nonlinear model for the first stage
and use the residuals from that to construct a control function, but
this requires stronger assumptions (Wooldridge, 2015). Instead, because
$x_{it}$ is observed rather than latent, you can just use $x_{it}$
as a first-stage dependent variable rather than $\boldsymbol{1}(x_{it}>0)$.
Then, construct a control function as normal. For an example of this,
see Palmer (2013, 2023), who does this to instrument an indicator
for a loan being underwater. A loan is underwater if its loan-to-value
ratio exceeds 100\%, so the loan-is-underwater indicator can be instrumented
by using loan-to-value ratio as a first-stage dependent variable.

\subsubsection{Linear versus Nonlinear Auxiliary Models\label{subsubsec:theory-nonlinearaux}}

It is a well-known fact that we can always run a linear first stage
if the second stage is linear (Angrist \& Pischke, 2009; Kelejian,
1971). Even if the true auxiliary model (``first stage'') is nonlinear,
assuming that the true primary model (``second stage'') is linear,
under typical 2SLS assumptions, the resulting estimates will still
be consistent. This is why running a nonlinear first stage in such
a scenario is considered a ``forbidden regression'' (Angrist \&
Pischke, 2009; Hausman, 1975) -- if you did so, the resulting estimates
would be much less robust to misspecification since their consistency
is only guaranteed under strong assumptions. Unfortunately, when the
primary model is nonlinear, this result no longer applies. Given that
the second stage is cloglog, this means we must carefully consider
the form of the first stage.

Using a linear first stage is not necessary if the first-stage dependent
variables are continuous. So, why is this assumed in the Section \ref{subsec:theory-modelintro}
model? Well, as with all two-stage IV methods, we want an estimator
of the conditional expectation function of the first stage in order
to exploit some sort of orthogonality condition (here, conditional
independence) between the instruments and the second-stage error term.
OLS is an estimator with the most parsimonious functional form: OLS
provides the best linear approximation to the conditional expectation
function. So, even if the true first stage is not truly linear, OLS
will often provide a good approximation to reality when the first-stage
dependent variables are continuous.

On the other hand, if the first-stage dependent variables are discrete,
then we know that a linear (i.e., a linear probability model) is a
poor approximation to reality if the predicted probabilities are substantially
different from 0.5. This motivates the use of a nonlinear first stage.
Theorem \ref{thm:CAN} can easily be extended to allow for this: because
the form of the first stage is not central to the proof (see Appendix
\ref{sec:proof}), a nonlinear first stage is fine in principle --
we just need to make additional assumptions. For more detail, see
Wooldridge's (2015) work on the control function approach for discrete
first-stage dependent variables.

\section{Stata Implementation\label{sec:stata}}

\subsubsection*{An Automated Stata Command}

The GMM estimation procedure described in Section \ref{sec:theory}
can be easily and conveniently implemented using the Stata command
\texttt{\textbf{\small{}ivcloglog}}, which is available from Boston
College's SSC (Statistical Software Components archive). As with all
commands available through SSC, Stata users can install \texttt{\textbf{\small{}ivcloglog}}
by simply running \texttt{\textbf{\small{}ssc install ivcloglog}}
in Stata. (\texttt{\textbf{\small{}ssc}} is a built-in Stata command.)

Before covering how to estimate the instrumented version of the Prentice
and Gloeckler (1978) model, I will first give an overview of how to
estimate the original, basic version (i.e., with all regressors exogenous).
Since the Prentice and Gloeckler (1978) model can be nominally expressed
as a cloglog model with time fixed effects (Allison, 1982; Jenkins,
1995), the setup required for estimating it in Stata is very simple:
\begin{enumerate}
\item If not the case already, format the data so that the observations
for each entity $i$ end after event occurrence.
\item If not the case already, generate the outcome variable -- a binary
indicator for event occurrence that is 1 in the time period of the
event and 0 otherwise.
\item Using the resulting dataset, simply estimate a cloglog model (with
the built-in command \texttt{\textbf{\small{}cloglog}}) using the
aforementioned outcome variable and with fixed effects for each time
period included as controls.
\end{enumerate}
(For a user-created command, \texttt{\textbf{\small{}pgmhaz8}}, that
automates these steps and estimates a model with exogenous regressors,
see Jenkins, 2004.) The process for estimating a model with endogenous
regressors is identical, except with the use of the command \texttt{\textbf{\small{}ivcloglog}}
instead of \texttt{\textbf{\small{}cloglog}}.

\texttt{\textbf{\small{}ivcloglog}} checks for collinear variables
(in the instruments $Z$ and augmented second-stage regressors $\widehat{\Xi}$)
or perfect predictors (in the augmented second-stage regressors $\widehat{\Xi}$)
and excludes them from the estimation. Using the remaining variables,\texttt{\textbf{\small{}
ivcloglog}} obtains the point estimates using \texttt{\textbf{\small{}regress}}
(for the first-stage) and \texttt{\textbf{\small{}cloglog}} (for the
second-stage) and then calculates an appropriate variance-covariance
matrix estimate (VCE) that covers both the first and second stages
using the formulae detailed in Section \ref{sec:theory}. (The OLS
first stage means that all endogenous regressors are assumed to be
either continuous or a known, possibly discontinuous function of the
first-stage dependent variables.) By default, the VCE is simply obtained
through the \texttt{\textbf{\small{}gmm}} command.

However, Stata may fail to do so through the \texttt{\textbf{\small{}gmm}}
command due to numerical difficulties. In particular, although the
earlier collinearity-checking guarantees that $\widehat{G}$ is exactly
invertible (see Part 2B), if $\widehat{G}$ is highly near-singular,
Stata may nonetheless declare that $\widehat{G}$ is numerically singular
and refuse to obtain the VCE $\widehat{V}=\widehat{G}^{-1}\widehat{\Omega}(\widehat{G}^{-1})'$.\footnote{Technically, it is more efficient to find $\widehat{V}$ by using
LU decomposition to solve $\widehat{G}\widehat{V}\widehat{G}'=\widehat{\Omega}$
for $\widehat{V}$ rather than directly calculating $\widehat{V}=\widehat{G}^{-1}\widehat{\Omega}(\widehat{G}^{-1})'$,
and this is presumably what Stata actually does. In any case, Stata
will not obtain $\widehat{V}$.} For example, this can occur when the dataset has many observations
and is very sparse. The presence of sparse columns in $Z$ and $\widehat{\Xi}$
means that changing the corresponding fixed effect parameters affects
the moment functions very little, resulting in near-zero columns in
the (sample) matrix of moment function derivatives $\widehat{G}$.
Stata can treat these as actually being columns of zeros.

This sparsity issue can be particularly common when estimating any
version of the Prentice and Gloeckler (1978) model due to the need
to include time fixed effects, and can be compounded if other fixed
effects are included as controls. This is because each level of a
fixed effect represents a dummy variable (i.e., a binary indicator
variable), and we must directly include these dummy variables as regressors
-- since the second stage is nonlinear rather than linear, you cannot
simply instead ``absorb'' the fixed effects by applying a within-transformation
to the whole dataset.

The option \texttt{\textbf{\small{}difficult\_vce}}, which requests
that alternative code be used to calculate the VCE, resolves this
problem. With this option, the code will construct $\widehat{G}$
and $\widehat{\Omega}$ and then solve for $\widehat{V}$ with a tolerance
threshold of 0 (i.e., floating-point numbers are treated as 0 only
if they are exactly 0). Since $\widehat{G}$ is exactly invertible,
this prevents Stata from declaring that $\widehat{G}$ is numerically
singular.

An alternative, fast procedure for obtaining the second-stage VCE
for generic control function models is described by Terza (2016).

\subsubsection*{}

\section{Conclusion\label{sec:concl}}

This paper introduces a theory-complete IV methodology for the estimation
of hazard models and discrete choice models) that have endogenous
regressors, including when there is measurement error in the regressors.
Although there are existing methods for dealing with mismeasured regressors
in hazard models, these are not ``true'' IV methods: they do not
allow for other types of endogenous regressors, and are therefore
limited in their usefulness. On the other hand, there do exist proposals
on IV methodologies to tackle the problem of endogeneity in hazard
models, but none give an adequate treatment of the subject: these
are either highly heuristic and not theoretically grounded or dangerously
unrobust. Similarly, existing papers on the use of control functions
to instrument discrete choice models are narrowly restricted to probit
and Poisson regression, and a solution for generic discrete choice
models is missing.

To fill this gap, I rigorously demonstrate how the IV relevance condition
and instrument conditional independence assumption (the control function
analog of the 2SLS exclusion restriction), together with minor conditions
and the assumption of a sufficiently flexible functional form for
the control function, allow causal effects to be recovered from hazard
models and discrete choice models with generic endogenous regressors.

This paper contributes to a nascent literature, and theoretically
justifies and undergirds already existing empirical work that uses
the proposed estimator. It supplies a theory guide, thus serving as
a crutch for existing and future empirical work, and this is compounded
by the guidance, presented in this paper, on my Stata command, \texttt{\textbf{\small{}ivcloglog}},
which allows practitioners to easily implement the proposed estimator.
Furthermore, this paper serves as a starting point for future theoretical
work to further improve the robustness of IV methodologies in hazard
models. All in all, this paper opens up promising new theoretical
and empirical avenues for practitioners to investigate.

\newpage{}

\begin{singlespace}
\bibliographystyle{hapalike_mod}
\nocite{*}
\bibliography{ivhazard_theory_guide_v1}
\newpage{}
\end{singlespace}

\appendix
\numberwithin{equation}{section}
\setcounter{equation}{0}
\setcounter{figure}{0} \renewcommand{\thefigure}{A\arabic{figure}}
\setcounter{table}{0} \renewcommand{\thetable}{A\arabic{table}}
\setcounter{assumption}{0} \renewcommand{\theassumption}{A\arabic{assumption}}
\begin{singlespace}

\section*{Appendix}
\end{singlespace}

\section{\label{sec:proof}Proof of Theorem \ref{thm:CAN}}

\subsection*{A Finite-Dimensional Z-Theorem}

I will first introduce a version of Huber's (1967; Theorem 3) Z-theorem
based on Wellner (2010). The point of using this Z-theorem (non-classical
GMM) is to allow the use of the Lindeberg-Feller CLT.

\renewcommand{\theassumption}{H\arabic{assumption}}
\begin{assumption}
$\hat{\theta}$ is the solution of $\hat{m}(\hat{\theta})\coloneqq\frac{1}{n}\stackrel[i=1]{n}{\sum}g_{i}(\hat{\theta})=o_{p}(n^{-\frac{1}{2}})$
and $\theta_{0}$ is the unique solution of $m(\theta)\coloneqq\mathbb{E}\left[g_{i}(\theta)\right]=0$.
\end{assumption}
\begin{assumption}
$\sqrt{n}\left(\hat{m}(\theta_{0})-m(\theta_{0})\right)\xrightarrow{d}D_{0}$,
where $D_{0}$ is some random variable that has a known distribution.
\end{assumption}
\begin{assumption}
For some norm $\left\Vert \cdot\right\Vert $ equivalent to the Euclidean
norm, for every sequence $\delta_{n}\rightarrow0$,
\[
\underset{\left\Vert \theta-\theta_{0}\right\Vert \leq\delta_{n}}{\sup}\frac{\left\Vert \sqrt{n}\left(\hat{m}(\theta)-m(\theta)\right)-\sqrt{n}\left(\hat{m}(\theta_{0})-m(\theta_{0})\right)\right\Vert }{1+\sqrt{n}\left\Vert \theta-\theta_{0}\right\Vert }\xrightarrow{p}0.
\]
\end{assumption}
\begin{assumption}
The function $m$ is Fr\'echet-differentiable at $\theta_{0}$ with
non-singular gradient
\[
G\coloneqq\nabla_{\mkern-4mu \theta}m(\theta)\equiv\left.\frac{\partial m(\theta)}{\partial\theta'}\right|_{\theta=\theta_{0}}.
\]
\end{assumption}
\begin{assumption}
The true finite-dimensional parameter $\theta_{0}$ is in the interior
of the parameter space.\footnote{This assumption is not explicitly stated in Wellner (2010), but I
include it here. The reason for assuming that $\theta_{0}$ is in
the interior of the parameter space is that, in Huber's (1967) Theorem
3, the Z-theorem only gives results for an arbitrary open strict subset
of the original parameter space.}
\end{assumption}
\begin{manualtheorem}{H1}

Suppose Assumptions H1-H5 hold. Then,
\[
\sqrt{n}(\hat{\theta}-\theta_{0})\xrightarrow{d}-G^{-1}D_{0}.
\]
In particular, if $D_{0}\sim N(0,\Omega)$, we have
\[
\sqrt{n}(\hat{\theta}-\theta_{0})\xrightarrow{d}N(0,G^{-1}\Omega(G^{-1})').
\]

\end{manualtheorem}

\subsection*{A Version of the Leibniz Integral Rule for Expectations}

Next, I introduce a special case of the measure theory variant of
the Leibniz Integral Rule\emph{.}
\begin{lem}
\label{lem:interchange}Let $X\in\mathcal{X}$ be a random variable
and $f\colon\mathbb{R}\times\mathcal{X}\rightarrow\mathbb{R}$ a function
such that $f(\theta,X)$ is integrable for all $\theta$ and $f$
is differentiable w.r.t.\ $\theta$. If there exists a random variable
$Z$ such that $\left|\frac{\partial}{\partial\theta}f(\theta,X)\right|\leq Z$
for all $\theta$ with $\mathbb{E}\left[Z\right]<\infty$, then

\[
\mathbb{E}\left[\frac{\partial}{\partial\theta}f(\theta,X)\right]=\frac{\partial}{\partial\theta}\mathbb{E}\left[f(\theta,X)\right].
\]
\end{lem}
\begin{proof}
This follows from a straightforward application of the Dominated Convergence
Theorem. For more detail, see the proof of Theorem 3.5.1 in Norris
(n.d.; lecture notes); Lemma \ref{lem:interchange} is a direct corollary
of Theorem 3.5.1.
\end{proof}

\subsection*{Main Proof}

This proof will use the notation and the model introduced in Section
\ref{subsec:theory-modelintro}.

\subsubsection*{Part 1 -- Verification of the Z-Theorem Conditions}

Given the auxiliary model, $u_{it}\indep z_{it}|v_{it}$ (Assumption
\ref{assu:condindep}) implies that $(x_{it},z_{it})\indep u_{it}|v_{it}$,
i.e., that $x_{it}$ is endogenous only through $v_{it}$. Using the
decomposition of $u_{it}$, we then have that $(x_{it},z_{it})\indep e_{it}|v_{it}$.
We also have that $(x_{it},z_{it})\indep e_{it}|p(v_{it})$; we can
equivalently condition on $p(v_{it})$ instead of $v_{it}$ because
$p(v_{it})$ depends only on $v_{it}$ and contains $v_{it}$ as one
of its elements.

Now, $e_{it}|p(v_{it})\equiv e_{it}|v_{it}$, and therefore $e_{it}|p(v_{it})\sim\mathrm{Gumbel}(0,1)$,
because $p(v_{it})$ depends only on $v_{it}$ and contains $v_{it}$
as one of its elements. Therefore, using the fact that $e_{it}|v_{it}\sim\mathrm{Gumbel}(0,1)$
(Assumption \ref{assu:CFspec}), we have that $e_{it}|p(v_{it})\sim\mathrm{Gumbel}(0,1)$
. We can then combine this with the previous result that $(x_{it},z_{it})\indep e_{it}|p(v_{it})$
to yield
\[
e_{it}|x_{it},z_{1it},p(v_{it})\sim\mathrm{Gumbel}(0,1).
\]

This means that the primary model is a cloglog model when the control
function terms $p(v_{it})$ are included alongside the original second-stage
regressors $x_{it}$ and $z_{1it}$. Therefore, this implies that
including the terms of $p(v_{it})$ as controls allows us to recover
an exogenous Prentice and Gloeckler (1978) model because the base
assumptions for it are satisfied (Assumptions \ref{assu:setup} and
\ref{assu:regularity} as well the primary model rank condition of
Assumption \ref{assu:rank}). Consequently, if $p(v_{it})$ were known,
we could just do that to get consistent estimates of $\beta_{1}$
and $\beta_{2}$, but we cannot because $v_{it}$ is unobserved. Instead,
if we have an estimate $\hat{\pi}$ of the auxiliary model parameters,
we can consider substituting $v_{it}$ with $\hat{v}_{it}\coloneqq x_{it}-z_{it}'\hat{\pi}$
when estimating. To see why, first substitute the unobserved variable
$v_{it}$ for the observed variables in $z_{it}$ and unknown parameter
vector $\pi$ using the auxiliary model. This gives us
\[
y_{it}=\mathbf{1}(\psi_{t}+\mathfrak{z}_{it}'\beta_{1}+\beta_{2}x_{it}+p(x_{it}-z_{it}'\pi)'\beta_{3}+e_{it}>0).
\]
This is a classic two-step Z-estimation setup: were it not for the
unknown parameter vector $\pi$, we would be able to directly estimate
the remaining parameters using MLE. We want to replace $\pi$ with
$\hat{\pi}$ and do QMLE instead. Well, as shown in Section \ref{subsec:theory-modelintro},
the moment functions can be stacked, so we end up with a one-step
Z-estimation problem. Therefore, we just need to apply a Z-theorem
to find sufficient conditions and get results; this proof will use
Theorem H1.

The first-stage moment conditions are just OLS moment conditions,
so they will not be discussed further -- the assumptions of Theorem
H1 trivially hold for the first stage. I will instead discuss how
the second-stage moment conditions satisfy these Z-theorem conditions.
Recall that $\theta\equiv(\pi',\psi',\beta')'$ and that the moment
conditions are
\[
\hat{m}(\theta)\coloneqq\left(\begin{array}{c}
\hat{m}_{1}(\theta)\\
\hat{m}_{2}(\theta)
\end{array}\right)\equiv\left(\begin{array}{c}
\frac{1}{n}\stackrel[i=1]{n}{\sum}g_{1i}(\theta)\\
\frac{1}{n}\stackrel[i=1]{n}{\sum}g_{2i}(\theta)
\end{array}\right)\equiv\frac{1}{n}\stackrel[i=1]{n}{\sum}g_{i}(\theta),
\]
where
\[
\hat{m}_{1}(\theta)\equiv\hat{m}_{1}(\pi)\coloneqq\frac{1}{n}\frac{\partial l_{1}}{\partial\pi}\equiv\frac{1}{n}\stackrel[i=1]{n}{\sum}\stackrel[t=1]{T}{\sum}z_{it}'v_{it}\equiv\frac{1}{n}\stackrel[i=1]{n}{\sum}\underset{g_{1i}(\theta)}{\underbrace{\stackrel[t=1]{T}{\sum}z_{it}'(x_{it}-z_{it}'\pi)}}
\]
and
\[
\hat{m}_{2}(\theta)\equiv\hat{m}_{2}(\psi,\beta,\pi)\coloneqq\frac{1}{n}\left(\begin{array}{c}
\frac{\partial l_{2}}{\partial\psi_{1}}\\
\vdots\\
\frac{\partial l_{2}}{\partial\psi_{T}}\\
\frac{\partial l_{2}}{\partial\beta_{1}}\\
\vdots\\
\frac{\partial l_{2}}{\partial\beta_{K}}
\end{array}\right)\equiv\frac{1}{n}\left(\begin{array}{c}
\frac{\partial}{\partial\psi_{1}}\stackrel[i=1]{n}{\sum}l_{2i}\\
\vdots\\
\frac{\partial}{\partial\psi_{T}}\stackrel[i=1]{n}{\sum}l_{2i}\\
\frac{\partial}{\partial\beta_{1}}\stackrel[i=1]{n}{\sum}l_{2i}\\
\vdots\\
\frac{\partial}{\partial\beta_{K}}\stackrel[i=1]{n}{\sum}l_{2i}
\end{array}\right)\equiv\frac{1}{n}\stackrel[i=1]{n}{\sum}\underset{\mathclap{\coloneqq g_{2i}(\theta)\equiv\frac{1}{T}\stackrel[t=1]{T}{\sum}g_{2it}(\theta)}}{\underbrace{\left(\begin{array}{c}
\frac{\partial l_{2i}}{\partial\psi_{1}}\\
\vdots\\
\frac{\partial l_{2i}}{\partial\psi_{T}}\\
\frac{\partial l_{2i}}{\partial\beta_{1}}\\
\vdots\\
\frac{\partial l_{2i}}{\partial\beta_{K}}
\end{array}\right)}},
\]
where $n$ is the number of entities (indexed by $i$), $T$ is the
number of time periods, and $K$ is the number of regressors. The
elements of $g_{2i}(\theta)$ are defined by
\[
\frac{\partial l_{2i}}{\partial\psi_{t}}=\begin{cases}
-h_{it} & t<s\\
\delta_{i}b_{it} & t=s\\
0 & t>s
\end{cases};\qquad\frac{\partial l_{2i}}{\partial\beta_{k}}=\mathfrak{\hat{\zeta}}_{is,k}\left(\delta_{i}b_{is}-\stackrel[j=1]{s-1}{\sum}h_{ij}\right),
\]
where $\delta_{i}$ is the censoring indicator, $t$ represents a
time period, $s$ is the time period of survival (with subscript $i$
suppressed for clarity), $\mathfrak{\hat{\zeta}}_{it,k}$ is the $k$-th
element of $\mathfrak{\hat{\zeta}}_{it}$ (i.e., the observation for
the $k$-th regressor, excluding the time fixed effects, for entity
$i$ and time $t$$)$, and where
\[
b_{is}\coloneqq\frac{h_{it}e^{-h_{it}}}{1-e^{-h_{it}}};\qquad h_{it}\coloneqq e^{\psi_{t}+\mathfrak{\zeta}_{it}'\beta}.
\]
$\hat{m}_{1}(\theta)$ is a score function, but $\hat{m}_{2}(\theta)$
is technically a quasi-score function because $p(\hat{v}_{it})$ is
a vector of generated regressors. $g_{1i}(\theta)$ is the entity-specific
first-stage score contribution,$g_{2i}(\theta)$ is the entity-specific
second-stage quasi-score contribution, and $g_{2it}(\theta)$ is the
observation-specific second-stage quasi-score contribution; $g_{i}(\theta)\coloneqq(g_{1i}(\theta)',g_{2i}(\theta)')'$.
Recall that $\hat{\theta}\equiv(\hat{\pi}',\hat{\psi}',\hat{\beta}')'$
is the solution of
\[
\hat{m}(\hat{\theta})=0.
\]

Before continuing on to prove the assumptions of Theorem H1, first
note that Assumption \ref{assu:regularity} implies that $\mathbb{E}\left[|g_{i}(\theta)|\right]<\infty$.
This is a consequence of the instrument rank condition (for the first-stage
components) and the Lindeberg condition (for the second-stage components).
This is a useful fact that will used more than once.

To prove Assumption H1, we must prove the uniqueness of $\theta_{0}$
and $\hat{\theta}$ and that $m(\theta_{0})=0$. The former just directly
follows from the monotonicity of $m$ and $\hat{m}$, but proving
$m(\theta_{0})=0$ is a little more involved. Trivially, $m_{1}(\theta_{0})=0$
because the first stage is just OLS. For $m_{2}(\theta_{0})$, we
just need to (1) note that the base conditions for Prentice and Gloeckler
(1978) are satisfied and (2) prove that you can still interchange
the integral and partial-derivative operators contained within $m_{2}(\theta)\equiv\mathbb{E}\left[g_{2i}(\theta)\right]$.
For the former, we have Assumptions \ref{assu:setup} and \ref{assu:regularity}
as well the primary model rank condition of Assumption \ref{assu:rank}.
As for the latter, this can easily be done by verifying the conditions
for Lemma \ref{lem:interchange} by noting that $l_{2i}$ is integrable,
$g_{2i}(\theta)$ exists, and $|g_{2i}(\theta)|$ is bounded over
the parameter space by a random variable with finite expectation (because
$\mathbb{E}\left[|g_{i}(\theta)|\right]<\infty$).\footnote{\label{fn:CIMEproof}Similarly, it is easy to show that an analog
of the information matrix equality (see Section \ref{subsubsec:theory-effandCIME})
holds here.}

Assumption H2 holds because the Lindeberg condition (on the individual-specific
second-stage quasi-score contributions) allows the Lindeberg-Feller
CLT to be used. In particular, the Lindeberg-Feller CLT implies that
$\hat{m}_{2}(\theta_{0})$ is asymptotically normal and has known
asymptotic mean and asymptotic variance.\footnote{Prentice and Gloeckler (1978) also use the Lindeberg condition in
order to use the Lindeberg-Feller CLT: ``The basic requirement is
that the individual information matrices not be too disparate among
study subjects in order that the central limit theorem apply to the
distribution of {[}the score{]}.''} (The asymptotic mean is $m(\theta_{0})=0$.)

For Assumption H3, a sufficient condition is that $\hat{m}$ be uniformly
convergent in probability, which can be shown using a uniform weak
law of large numbers (Theorem 12.1 in Wooldridge, 2010). The compactness
of the parameter space (together with the Borel-measurability and
continuity of $\hat{m}$) means that we just require $|\hat{m}|$
to be bounded over the parameter space by a random variable with finite
expectation, but this is trivially implied by the fact that $\mathbb{E}\left[|g_{i}(\theta)|\right]<\infty$.

Assumption H4 simply requires that the expected Hessians of the first
and second stages,\footnote{Here, OLS is cast as a special case of MLE for ease of exposition.}
evaluated at the true parameter, must be invertible -- this is implied
by the rank conditions. (This is shown in Part 2A.) Assumption H5
is directly assumed as part of Assumption \ref{assu:regularity}.

Satisfying all assumptions of Theorem H1, the properties of $\hat{\theta}$
follow.\footnote{Because the variance follows the sandwich formula (this is because
Theorem H1 is based on the delta method), it is obvious that the standard
results on sequential (classical) GMM estimation (Newey, 1984; Newey
\& McFadden, 1994) also carry over, despite the use of non-classical
GMM.} All that is left is to derive the specific formulae for the asymptotic
variance of $\hat{\theta}$ (Part 2A) and a consistent estimator for
the asymptotic variance of $\hat{\theta}$ (Part 2B).

\subsubsection*{Part 2A -- The Asymptotic Variance}

From applying Theorem H1 in Part 1, we know that\footnote{An aside: for this model, it is true that $\nabla_{\mkern-4mu \theta}\,\mathbb{E}\left[g_{i}(\theta_{0})\right]=\mathbb{E}\left[\nabla_{\mkern-4mu \theta}\,g_{i}(\theta_{0})\right]$.
Using Lemma \ref{lem:interchange}, it is easy to show that we can
interchange the integral and second-partial-derivative operators.}
\[
\sqrt{n}(\hat{\theta}-\theta_{0})\xrightarrow{d}N(0,G^{-1}\Omega(G^{-1})'),
\]
where
\[
G\equiv\mathbb{E}\left[\nabla_{\mkern-4mu \theta}\,g_{i}(\theta_{0})\right];\qquad\Omega\equiv\mathbb{E}\left[g_{i}(\theta_{0})g_{i}(\theta_{0})'\right].
\]

For completeness, note that the variance-covariance matrix here, $V\equiv G^{-1}\Omega(G^{-1})'$,
can be derived from the general GMM variance formula. Because we are
doing just-identified GMM, $G$ is square, and therefore
\begin{align*}
(G'WG)^{-1}G'W\varOmega W'G(G'W'G)^{-1} & =G^{-1}W^{-1}(G^{-1})'G'W\varOmega W'GG^{-1}(W^{-1})'(G^{-1})'\\
 & =G^{-1}\Omega(G^{-1})'.
\end{align*}
The LHS is the general formula, where $W$ is the positive semi-definite
GMM weight matrix.\footnote{This general formula is a standard result that can be obtained from
a Z-theorem by defining $m(\theta)$ in a certain way. Also, notice
that the choice of weight matrix for just-identified GMM is arbitrary:
it does not affect the end result. This should be unsurprising --
$W$ weights the individual moment functions in $\hat{m}(\theta)$
in the sense that, if it is not possible to set all of them to 0,
it represents which moment functions are prioritized. In just-identified
GMM, there are exactly as many instruments as individual moment functions,
so it is possible to set the entirety of $\hat{m}(\theta)$ to 0.}

Since we already know what $g_{i}(\theta)$ looks like, we know what
$\Omega$ looks like; however, it may not be immediately obvious to
the reader how to differentiate $g_{i}(\theta)$ to derive the structure
of
\[
G\equiv\left[\begin{array}{c|c}
G_{11} & G_{12}\\
\hline G_{21} & G_{22}
\end{array}\right].
\]
Thankfully, we can utilize the equivalent cloglog representation of
the model to express$g_{2i}(\theta)$ in a easier-to-handle form by
simply applying standard results for cloglog. 

$G\equiv\nabla_{\mkern-4mu \theta}m(\theta)$ is the gradient of $m(\theta)$
w.r.t.\ $\theta$, so $G_{ab}$ is the gradient of the $a$-th-stage
moment functions w.r.t.\ the $b$-th-stage parameters. $G_{11}$
and $G_{22}$ can be directly obtained as standard results from the
OLS and cloglog models. The form of $G_{11}$ is trivial. For $G_{22}$,
note that the observation-specific quasi-score contribution $g_{2it}(\theta)$
has the form of a cloglog model score contribution:
\[
g_{2it}(\theta)=\begin{cases}
\frac{e^{\xi_{it}'\gamma-e^{\xi_{it}'\gamma}}}{1-e^{-e^{\xi_{it}'\gamma}}}\xi_{it} & y_{it}=1\\
-e^{\xi_{it}'\gamma}\xi_{it} & y_{it}=0
\end{cases}.
\]
We can simply differentiate this to derive $G_{22}$. For $G_{12}$,
recall that we are dealing with a triangular system -- the second-stage
parameters do not show up in the first stage, so $G_{12}$ is trivially
just a matrix of zeroes. Finally, we can obtain the form of $G_{21}$
using the chain rule and the previous result for $G_{22}$ -- simply
differentiate the second-stage quasi-score w.r.t.\ $p(v)'\beta_{3}$,
differentiate $p(v)'\beta_{3}$ w.r.t.\ $\pi$, and take their product.
(This is why $\mkern3mu\mathchar'26\mkern-12mu d_{it}$ is just $d_{it}$
multiplied by the factor $\sum_{q=1}^{Q}-\beta_{3q}v_{it}^{q-1}$.)

Proving that $G$ is invertible is straightforward. Since $G$ is
a block lower triangular matrix, we know that $G$ has a rank of at
least its diagonal sub-matrices: $\mathrm{rank}(G)\geq\mathrm{rank}(\mathbb{E}\left[-z_{it}z_{it}'\right])+\mathrm{rank}(\mathbb{E}\left[d_{it}\xi_{it}\xi_{it}'\right])$.
$\mathbb{E}\left[-z_{it}z_{it}'\right]$ has full rank due to the
instrument rank condition; $\mathbb{E}\left[d_{it}\xi_{it}\xi_{it}'\right]$
has full rank because $\mathbb{E}\left[\xi_{it}\xi_{it}'\right]$
has full rank (the primary model rank condition) iff $\mathbb{E}\left[d_{it}\xi_{it}\xi_{it}'\right]$
does. This last claim can be easily proven by noting that $\xi_{it}\xi_{it}'$
is of quadratic form and that $d_{it}<0$.

First, let $v\in\mathbb{R}^{L_{\Xi}},v\neq0$ be given. $\xi_{it}\xi_{it}'$
is positive semi-definite because of its quadratic form, which implies
that
\[
v'\mathbb{E}\left[d_{it}\xi_{it}\xi_{it}'\right]v\leq0;\qquad v'\mathbb{E}\left[\xi_{it}\xi_{it}'\right]v\geq0,
\]
i.e., both matrices are semi-definite. (The fact that $d_{it}<0$
will cause the former to have an opposite sign from the latter, but
this is not the point.) Now, notice that $v'\mathbb{E}\left[d_{it}\xi_{it}\xi_{it}'\right]v\equiv\mathbb{E}\left[d_{it}(v'\xi_{it})^{2}\right]$
and $v'\mathbb{E}\left[\xi_{it}\xi_{it}'\right]v\equiv\mathbb{E}\left[(v'\xi_{it})^{2}\right].$
We have that

\[
\mathbb{E}\left[d_{it}(v'\xi_{it})^{2}\right]\neq0\iff\Pr((v'\xi_{it})^{2}\neq0)>0\iff\mathbb{E}\left[(v'\xi_{it})^{2}\right]\neq0
\]
since $(v'\xi_{it})^{2}\geq0$ and $d_{it}<0$. That is, because $(v'\xi_{it})^{2}$
is weakly positive, the only way for either expectation to be zero
for some $v\in\mathbb{R}^{L_{\Xi}}$ is if $(v'\xi_{it})^{2}=0$ with
probability one. But, combining this fact with the previous, this
means that
\[
v'\mathbb{E}\left[d_{it}\xi_{it}\xi_{it}'\right]v<0\iff v'\mathbb{E}\left[\xi_{it}\xi_{it}'\right]v>0,
\]
i.e., $\mathbb{E}\left[d_{it}\xi_{it}\xi_{it}'\right]$ is definite
(and thus invertible) iff $\mathbb{E}\left[\xi_{it}\xi_{it}'\right]$
is.

\subsubsection*{Part 2B -- A Consistent Estimator for the Asymptotic Variance}

By Slutsky, a consistent estimator of the asymptotic variance $V\equiv G^{-1}\Omega(G^{-1})'$
is $\widehat{V}\equiv\widehat{G}^{-1}\widehat{\Omega}(\widehat{G}^{-1})'$,
where $\widehat{G}$ is any consistent estimator of $G$ and $\widehat{\Omega}$
is any consistent estimator of $\Omega$. $\widehat{\Omega}$ is typically
just any standard GMM estimate of $\Omega$ given by statistical software,
such as when homoskedastic, ``robust'', or clustered standard errors
are requested. For $\widehat{G}$, we will use the sample version
of $G$ corresponding to $\hat{m}(\theta)$, which is guaranteed to
be consistent by Slutsky due to the continuity of $\widehat{G}$:
\[
\widehat{G}\equiv\left[\begin{array}{c|c}
\widehat{G}_{11} & \widehat{G}_{12}\\
\hline \widehat{G}_{21} & \widehat{G}_{22}
\end{array}\right].
\]
 Trivially, $\widehat{G}_{11}$ is the OLS variance for the first
stage, and $\widehat{G}_{12}=\boldsymbol{0}$. For $\widehat{G}_{21}$
and $\widehat{G}_{22}$, we can use the equivalent cloglog representation
of the Prentice and Gloeckler (1978) model and standard results for
the cloglog model; the structures of $\widehat{G}_{21}$ and $\widehat{G}_{22}$
in Theorem \ref{thm:CAN} immediately follow.

We can show that sample versions of the invertibility results from
Part 2A still hold. Using the same argument as before, since $\widehat{G}$
is block lower triangular, it is guaranteed to be invertible if $-Z'Z$
and $\widehat{\Xi}'\widehat{D}\widehat{\Xi}$ are. However, note that
$\widehat{\Xi}'\widehat{D}\widehat{\Xi}$ has the same rank as $\widehat{\Xi}'\widehat{\Xi}$
because $\widehat{D}$ is diagonal with non-zero diagonal elements.
Thus, the invertibility of $\widehat{G}$ is implied by the invertibility
of $Z'Z$ and $\widehat{\Xi}'\widehat{\Xi}$. Trivially, the instrument
and primary rank conditions imply that these will be invertible (and
therefore that the estimator proposed in this paper exists) with probability
approaching one as the sample size increases. For $\widehat{\Xi}'\widehat{\Xi}$,
this is because $\Xi'\Xi$ will be invertible with probability approaching
one and $\widehat{\Xi}'\widehat{\Xi}\xrightarrow{p}\Xi'\Xi$.\qed

\newpage{}

\section{\label{sec:primer}An Introduction to Grouped-Time Hazard Models}

This section will use different notation from the rest of the paper
that is closer to standard textbook notation.

First, let us consider a continuous-time single-event model. (Multi-event
models are outside the scope of this paper.) I will show how \emph{interval-censoring}
transforms a continuous-time model into a nominally discrete-time
model: a grouped-time model. (A discrete-time model can represent
a process that is fundamentally discrete-time -- e.g., the time variable
could be the number of cycles of some process, but it can also nominally
represent a grouped-time model.)

Let $T_{i}$, where $T_{i}\geq0$, denote the \textit{continuous time-to-event}
random variable, whose values are the time-to-event for a given entity
$i$ (individual, etc.). $T_{i}$ has PDF $f_{i}(\cdot)$ -- the
\textit{event density} -- and CDF $F_{i}(\cdot)$ -- the \textit{lifetime
distribution function}. The \textit{survival function} $S_{i}$ is
defined by
\[
S_{i}(t)\coloneqq\Pr(T_{i}>t)\equiv1-F_{i}(t)
\]
and is the probability of surviving up to time $t$. Assuming that
$S_{i}$ is everywhere continuously differentiable -- i.e., that
$f_{i}(\cdot)$ exists everywhere and is continuous, the derivative
of $S_{i}$ is the \textit{hazard function} $\lambda_{i}$, which
is defined by
\[
\begin{aligned}\lambda_{i}(t) & \coloneqq\underset{\eta\rightarrow0}{\lim}\frac{\Pr(t\leq T_{i}<t+\eta|T_{i}\geq t)}{\eta}\equiv\frac{-S_{i}'(t)}{S_{i}(t)}\equiv\frac{f_{i}(t)}{S_{i}(t)}\\
 & \equiv\underset{\eta\rightarrow0}{\lim}\mathbb{E}\left[\frac{\boldsymbol{1}(t\leq T_{i}<t+\eta)}{\eta}|T_{i}\geq t\right]
\end{aligned}
\]
and is the instantaneous rate of occurrence of the event, conditional
on the event not having occurred yet.\footnote{The last identity can be derived in a few steps. First, $\Pr(t\leq T_{i}<t+\eta|T_{i}\geq t)\equiv\mathbb{E}\left[\boldsymbol{1}(t\leq T_{i}<t+\eta)|T_{i}\geq t\right]\equiv\mathbb{E}\left[\boldsymbol{1}(T_{i}<t+\eta)|T_{i}\geq t\right]-\mathbb{E}\left[\boldsymbol{1}(T_{i}<t)|T_{i}\geq t\right]$.
Second, note that $\lim_{\eta\rightarrow0}\frac{\mathbb{E}\left[\boldsymbol{1}(T_{i}<t+\eta)|T_{i}\geq t\right]-\mathbb{E}\left[\boldsymbol{1}(T_{i}<t)|T_{i}\geq t\right]}{\eta}$
is a derivative; since $S_{i}$ is bounded and is everywhere differentiable
by assumption (note that the continuity and boundedness of $S_{i}$
imply integrability), by Lemma \ref{lem:interchange}, we can interchange
expectation and differentiation. Finally, we just need to rearrange
and then apply the Dominated Convergence Theorem to interchange the
limit and expectation operators.}

Note that the fact that $\lambda_{i}(t)\equiv\frac{-S_{i}'(t)}{S_{i}(t)}$
implies that
\[
\lambda_{i}(t)\equiv-\frac{\mathrm{d}}{\mathrm{d}t}\log S_{i}(t),
\]
which, after introducing the boundary condition $S_{i}(0)=1$ (i.e.,
we define $t=0$ to be some timepoint where no events have happened
yet), we can directly integrate by applying the First Fundamental
Theorem of Calculus to get
\[
S_{i}(t)\equiv\exp\Biggl(-\underset{\coloneqq\varLambda_{i}(t)}{\underbrace{\stackrel[0]{t}{\intop}\lambda_{i}(u)\,\mathrm{d}u}}\Biggr),
\]
where $\varLambda_{i}(t)$ is known as the \textit{integrated hazard}.
This expression for $S_{i}(t)$ is very intuitive -- the contribution
of the instantaneous hazard rate exponentially diminishes with time
because entities can ``fail'' before time $t$, after which the
hazard rate no longer matters for them. Let us now introduce the parameter
vector of interest $\beta$; suppose that $f_{i}$, $S_{i}$, and
$\lambda_{i}$ all depend on $\beta$. Let us also introduce right
censoring: rather than $T_{i}$, we observe the censored survival
time variable $Y_{i}=\min(T_{i},c_{i})$, where $c_{i}$ is the potential
censoring time. ($c_{i}$ is just when we no longer observe $T_{i}$,
i.e., the ``start'' of right censoring.) Let $\delta_{i}$ be an
indicator or whether entity $i$ gets right-censored, defined by
\[
\delta_{i}\coloneqq\begin{cases}
1 & T_{i}\leq c_{i}\\
0 & T_{i}>c_{i}
\end{cases}.
\]
Under mild assumptions (independence across individuals and non-random
censoring), the likelihood function is
\[
\begin{aligned}L_{n}(\beta;Y_{i},c_{i}) & =\stackrel[i=1]{n}{\prod}\biggl(\underset{f_{i}(Y_{i}|\delta_{i}=1)}{\underbrace{\frac{f_{i}(Y_{i};\beta)}{1-S_{i}(c_{i};\beta)}}}\times\underset{\Pr(\delta_{i}=1)}{\underbrace{(1-S_{i}(c_{i};\beta))\vphantom{\frac{f_{i}(Y_{i};\beta)}{1-S_{i}(c_{i};\beta)}}}}\biggr)^{\delta_{i}}\biggl(\underset{\mathclap{f_{i}(Y_{i}=0|\delta_{i}=0)}}{\underbrace{1\vphantom{\frac{f_{i}(Y_{i};\beta)}{1-S_{i}(c_{i};\beta)}}}}\times\underset{\mathrlap{\Pr(\delta_{i}=0)}}{\underbrace{S_{i}(c_{i};\beta)\vphantom{\frac{f_{i}(Y_{i};\beta)}{1-S_{i}(c_{i};\beta)}}}}\biggr)^{1-\delta_{i}}\\
 & =\stackrel[i=1]{n}{\prod}f_{i}(Y_{i};\beta)^{\delta_{i}}S_{i}(c_{i};\beta)^{1-\delta_{i}}
\end{aligned}
,
\]
where $n$ is the number of observations. The first $f_{i}$ refers
to a PDF and comes from the formula for a truncated distribution (since
$\delta_{i}=1$ means $T_{i}\leq c_{i}$). On the other hand, the
second $f_{i}$ instead refers to a PMF, and the fact that $f_{i}(Y_{i}=0|\delta_{i}=0)=\Pr(Y_{i}=0|\delta_{i}=0)=1$
directly results from the definitions of $Y_{i}$ and $\delta_{i}$.

Now let us instead consider a grouped-time scenario. We do not directly
observe times; instead, the time axis is partitioned into $R$ intervals
and we can only observe that an event falls into an interval, which
is known as ``interval censoring''. A grouped-time model involves
a special case of interval censoring where said intervals are identical.\footnote{That is, a grouped-time setting is when the data is split up into
the same time periods for all entities, which is typical of panel
data in economics. In contrast, a general interval-censored model
is one that has random interval censoring, which is more mathematically
complex. One classic example where such a model applies is medical
checkup data (since medical checkups do not occur at the same time
for all patients).} Let $A_{\tau}$ denote the $\tau$-th time interval (i.e., ``time
period $\tau$'') and let $a_{\cdot}$ denote a specific timepoint:
$A_{\tau}=[a_{\tau-1},a_{\tau})$, $\tau=1,\ldots,R$, with $a_{0}=0$
and $a_{K}=\infty$. Define the random variable $T_{i}^{*}$ to be
the discrete analog of $T_{i}$ -- it represents the time period
in which the event happens for entity $i$; similarly define $Y_{i}^{*}$.
Let $f_{i}^{*}$ denote the PMF of $T_{i}^{*}$, which is defined
by $f_{i}^{*}(\tau)\coloneqq\Pr(T_{i}\in A_{\tau})$; $f_{i}^{*}(\tau)$
is the integral of the PDF $f_{i}(\cdot)$ over time period $\tau$.
Obviously, the discrete CDF $F_{i}^{*}$ at time period $\tau$ is
just the sum of $f_{i}^{*}(\tau)$ over all time periods up to $\tau$,
and the corresponding survival function satisfies $S_{i}^{*}(\tau)\equiv1-F_{i}^{*}(\tau)$.
The discrete hazard function is the rate of failure in a given time
period, i.e., the conditional probability of failure:
\[
\lambda_{i}^{*}(\tau)\coloneqq\Pr(T_{i}\in A_{\tau}|T_{i}\geq a_{\tau-1})\equiv\frac{S_{i}^{*}(\tau-1)-S_{i}^{*}(\tau)}{S_{i}^{*}(\tau-1)}\equiv\frac{f_{i}^{*}(\tau)}{S_{i}^{*}(\tau-1)}.
\]
By the definition of conditional probability, it is clear that
\[
S_{i}^{*}(\tau)\equiv\stackrel[u=1]{\tau}{\prod}(1-\lambda_{i}^{*}(u)).
\]
Let $c_{i}^{*}$ and $\delta_{i}^{*}$ be the censoring time period
and censoring indicator respectively. (In panel data, $c_{i}^{*}$
is just the last time period for each entity $i$.) Similarly to before,
with right censoring, the likelihood function in the grouped-time
model is
\[
\begin{aligned}L_{n}^{*}(\beta;Y_{i}^{*},c_{i}^{*}) & =\stackrel[i=1]{n}{\prod}\biggl(\underset{f_{i}^{*}(Y_{i}^{*}|\delta_{i}^{*}=1)}{\underbrace{\frac{f_{i}^{*}(Y_{i}^{*};\beta)}{1-S_{i}^{*}(c_{i}^{*};\beta)}}}\times\underset{\Pr(\delta_{i}^{*}=1)}{\underbrace{(1-S_{i}^{*}(c_{i}^{*};\beta))\vphantom{\frac{f_{i}^{*}(Y_{i}^{*};\beta)}{1-S_{i}^{*}(c_{i}^{*};\beta)}}}}\biggr)^{\delta_{i}^{*}}\biggl(\underset{\mathclap{f_{i}^{*}(Y_{i}^{*}=0|\delta_{i}^{*}=0)}}{\underbrace{1\vphantom{\frac{f_{i}^{*}(Y_{i}^{*};\beta)}{1-S_{i}^{*}(c_{i}^{*};\beta)}}}}\times\underset{\mathrlap{\Pr(\delta_{i}^{*}=0)}}{\underbrace{S_{i}^{*}(c_{i}^{*};\beta)\vphantom{\frac{f_{i}^{*}(Y_{i}^{*};\beta)}{1-S_{i}^{*}(c_{i}^{*};\beta)}}}}\biggr)^{1-\delta_{i}^{*}}\\
 & =\stackrel[i=1]{n}{\prod}f_{i}^{*}(Y_{i}^{*};\beta)^{\delta_{i}^{*}}S_{i}^{*}(c_{i}^{*};\beta)^{1-\delta_{i}^{*}}
\end{aligned}
.
\]

It is typical to assume a basic hazard model where $f_{i}(t;\beta)=f(t|x_{i}'\beta)$
for some function $f$, i.e., where $f_{i}(t;\beta)$ depends on $i$
only through the \textit{single index} $x_{i}'\beta$, a linear combination
(of the values of) the regressors for entity $i$. $\lambda$, $F$,
and $S$ can be defined similarly. However, note that doing so implicitly
imposes the assumption that all entities are capable of experiencing
the event: if $F(\cdot)$ is a CDF, then $S(t)$ must tend to 0 as
$t$ goes to infinity.

Models that instead allow for some fraction of the population to be
incapable of experiencing the event (either ex-ante or becoming so)
are known as\emph{ cure models}. A \textit{non-mixture cure model}
models ``being cured'' as a competing (final) event in a competing
risks model: the same population faces multiple hazards (potentially
corresponding to different events or the same event). That is, the
focus is on the process of ``curing'' -- such models can be used
to model the rate at which patients get cured before dying, for example.
In contrast, a \emph{mixture cure model} models the cured and uncured
sub-populations as facing separate hazards. (The hazard rate could
be zero for the ``cured'' population, e.g., the hazard rate for
people being immune to some disease.) That is, the focus is on what
happens after ``curing'' -- this could be used to model the rate
at which patients die when they are cured of some disease, for example.
See Lambert et al.\ (2007) for an introduction to cure models.

Some care must be taken with the interpretation of an imposed basic
hazard model when the true underlying model is a cure model if some
entities truly never experience the event (rather than the occurrence
of the event being unobserved). Firstly, the hazard function and the
survival function are well-defined but not the time-to-event, which
will be an ``improper random variable'' (since its value is undefined
for some entities). Though, all three are well-defined if we condition
on the event eventually happening. Secondly, there is probable bias
and inconsistency due to misspecification even if cure status is completely
independent of all other variables -- moreover, in general, it is
not true that the estimate for $\beta$ converges to $\beta$ multiplied
by the uncured proportion.

\endgroup

\end{document}